\documentclass[12pnt,aps,pra,twocolumn,superscriptaddress]{revtex4-1}

\usepackage{ulem}
\usepackage{graphicx}  
\usepackage{dcolumn}   
\usepackage{bm}        
\usepackage{amssymb}
\usepackage{physics}   
\usepackage{mathtools}
\usepackage{esvect}
\usepackage{wrapfig}
 \usepackage[utf8]{inputenc}

\usepackage{amsthm}
\usepackage{ragged2e}
\usepackage{verbatim}
\usepackage[colorlinks=true,linkcolor=blue,citecolor=red,plainpages=false,pdfpagelabels]{hyperref}
\usepackage{color,xcolor,colortbl}
\usepackage{multirow}
\usepackage{boxedminipage}
\usepackage{changepage}
\usepackage{bbm}
\usepackage{lipsum}
\usepackage{array}
\usepackage{etoolbox}
\usepackage[caption=false]{subfig}
\usepackage{appendix}
\usepackage{float}

\usepackage[mathscr]{euscript}

\def\Tr{\operatorname{Tr}}

\def\SEP{\operatorname{SEP}}

\def\supp{\operatorname{supp}}

\def\>{\rangle}
\def\<{\langle}

\def\id{\operatorname{id}}

\def\({\left(}
\def\){\right)}
\def\[{\left[}
\def\]{\right]}

\newcommand{\mc}[1]{\mathcal{#1}}

\newtheorem{observation}{Observation}
\newtheorem{theorem}{Theorem}

\newtheorem{corollary}{Corollary}

\newtheorem{definition}{Definition}

\newtheorem{lemma}{Lemma}

\newtheorem{remark}{Remark}

\begin{document}

\widetext

%Fundamental limitations on device-independent quantum key distribution rates
\title{Upper bounds on device-independent quantum key distribution rates in static and dynamic scenarios}
\author{Eneet Kaur}\email{ e2kaur@uwaterloo.ca}
\affiliation{Institute for Quantum Computing and Department of Physics and Astronomy, University of Waterloo, Waterloo, Ontario N2L 3G1, Canada}
\author{Karol Horodecki}\email{karol.horodecki@ug.edu.pl}
\affiliation{National Quantum Information Centre in Gda\'{n}sk,
Faculty of Mathematics, Physics and Informatics, University of Gda\'{n}sk, 80–952 Gda\'{n}sk, Poland}
\affiliation{International Centre for Theory of Quantum Technologies, University of Gda\'nsk, Wita Stwosza 63, 80–308 Gdansk, Poland}
\affiliation{Institute of Informatics, University of Gda\'{n}sk}
\author{Siddhartha Das}\email{Siddhartha.Das@ulb.be}
\affiliation{Centre for Quantum Information \& Communication (QuIC), \'{E}cole polytechnique de Bruxelles,   Universit\'{e} libre de Bruxelles, Brussels, B-1050, Belgium}
\date{\today}

\begin{abstract}

In this work, we develop upper bounds for key rates for device-independent quantum key distribution (DI-QKD) protocols and devices. We study the reduced cc-squashed entanglement and show that it is a convex functional. As a result, we show that the convex hull of the currently known bounds is a  tighter upper bound on the device-independent key rates of standard CHSH-based protocol.
We further provide tighter bounds for DI-QKD key rates achievable by any protocol applied to the CHSH-based device.  This bound is based on reduced relative entropy of entanglement optimized over decompositions into local and non-local parts. In the dynamical scenario of quantum channels, we obtain upper bounds for device-independent private capacity for the CHSH based protocols. We show that the device-independent private capacity for the CHSH based protocols on depolarizing and erasure channels is limited by the secret key capacity of dephasing channels. 
\end{abstract}

\maketitle

The history of development of the quantum key distribution can be divided in
two stages. Security of the first protocols such as BB84~\cite{BB84} were based on the trust towards the manufacturer. The devices were assumed to be working according to their specification. The Eavesdropper was assumed only to interfere with the channel connecting the honest parties. In the second stage, taking its origins in Ekert's paper~\cite{E91} this assumption was dropped leading to the {\it device-independent} quantum cryptography. In parallel, the initial-- call it {\it device-dependent approach}-- was getting maturity. On practical side the point-to-point or relay-based QKD were achieved commercially and experimentally (see \cite{Horodecki_2021} and references therein). From a theoretical perspective, the limitations in form of upper bounds on the key rate were developed in various device-dependent scenarios~\cite{CW04,HHHO09,TGW14,Pirandola2017,WTB17} (also see~\cite{AH09,DBWH19}).

The device-independent (DI) approach got advanced meanwhile due to milestones both in theory \cite{ArnonFriedman2018} and experiment \cite{TMC+13,HBD+15}. Recently it was shown that some of the latter so called loophole-free experiments had small but non-zero key rate \cite{schwonnek2020robust}. It is therefore
time to ask if the key rate in this important cryptographic scenario needs to be that small. That is, to place upper bounds on the key rate in the device independent scenario.  Since seminal results of Ref.~\cite{AGM06} (in the case of the non-signaling adversary), till now, there is only a handful of papers tackling this problem \cite{KWW20,WDH19,CFH20,AL20,FBL+21,Eneetthesis} most of which concerns, as we do, the quantum adversary. In what follows we not only find tighter limitations, but also provide a unified view on the previous bounds exhibiting hidden connections. 

\begin{figure}
    \centering
    \includegraphics[width=\linewidth]{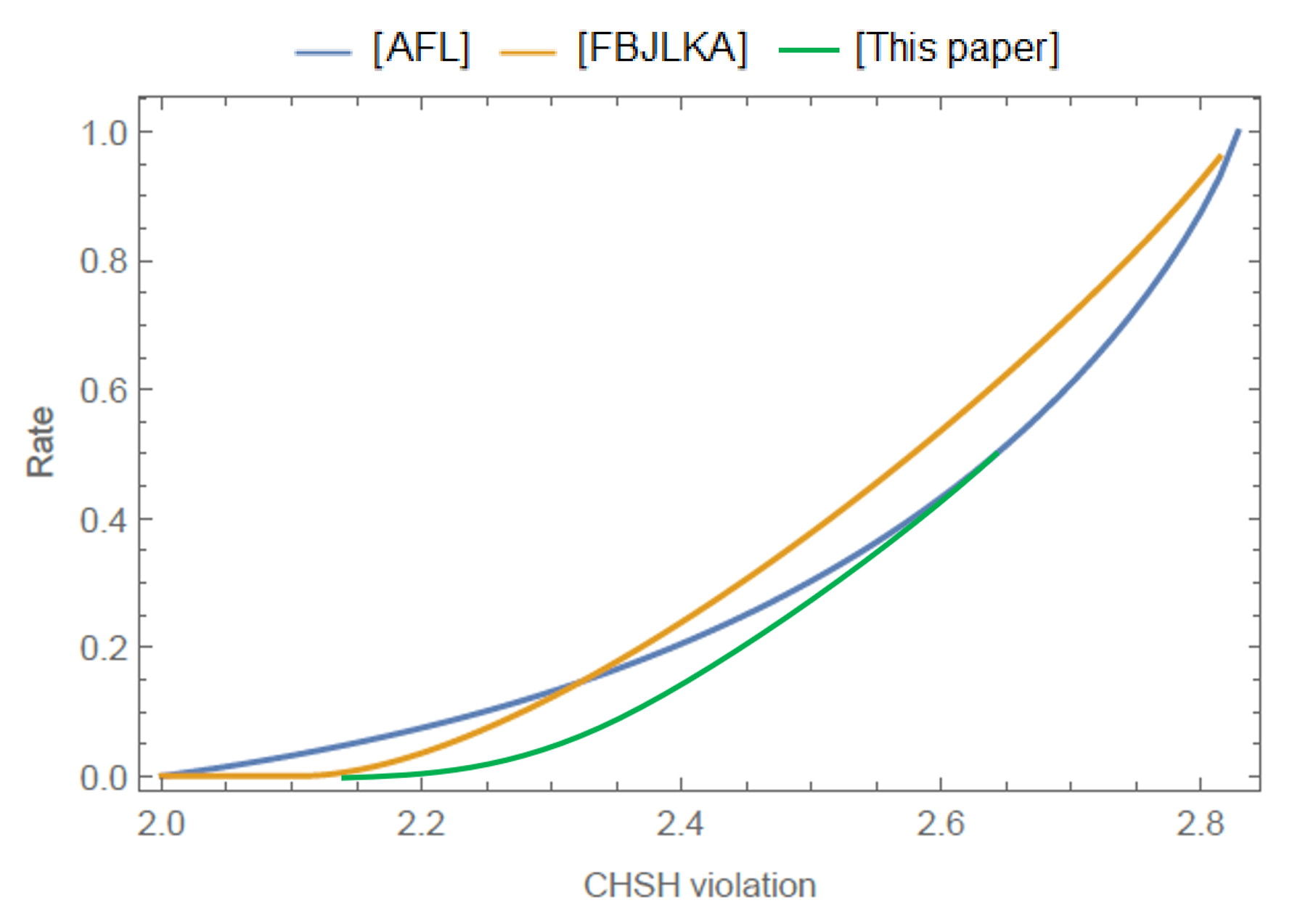}
    \caption{In this figure, we show the plots for standard device-independent CHSH protocol obtained in Refs.~\cite{AL20}, \cite{FBL+21}, and the upper bound given in Theorem~\ref{thm:convexification-m}, which is the convex hull of the former bounds, depicted in green.}
    \label{DI-QKD-Plots-m}
\end{figure}

After a seminal result of Ref.~\cite{KWW20} for the quantum adversary,
three approaches were taken: (i) that of Ref.~\cite{CFH20} where
the bound is proposed via reduced entanglement measures,
(ii) that of Ref.~\cite{AL20} where the intrinsic information is proposed as an upper bound via Ref.~\cite{CEHHOR} 
(iii) that of Ref.~\cite{FBL+21} where classical attack is proposed via the so called intrinsic information.
In Ref.~\cite{FBL+21}, a strong result was provided. Namely, certain  quantum states exhibiting non-locality, have {\it zero} quantum device-independent (QDI) key under standard protocols. The key rate considered there is obtained by protocols based on projective measurements and announcing publicly the inputs. It is easy to observe a direct analogy between this result and a previously obtained upper bound in a case when Eve is limited only by no-signaling communication~\cite{WDH19}. There the key achieved by a single measurement done in parallel is shown to be zero for certain quantum non-local devices (for any Bell inequality that can be used for testing). Noticing this connection will be crucial to our methods in going beyond FBJ\L{}KA bound.  

The most common DI-QKD protocols use only single measurement for key generation. In particular, its honest implementation is based on distributing the $2$-qubit Werner states (mixtures of a maximally entangled state with the maximally mixed state). The testing against the eavesdropper is based on the CHSH inequality \cite{CHSH}. The approaches of Refs.~\cite{AL20} and \cite{FBL+21} provide different upper bounds for such a protocol. AFL bound works in regime when the Werner state is close to the maximally entangled state, while the second works very well in the opposite regime-- when it is close to the maximally mixed state. This is because the first attack
is quantum (by a mixture of Bell states and tuned measurements) while the
second exploits errors, and works when such errors in the Werner state occurs. It was therefore not clear how to achieve a single bound which works in both regimes. In this work, we show that the two bounds are instances of the optimization of a single convex quantity. This allows us to obtain a bound that perform better than the above introduced bounds.

{\it Main results}.--- 
As the first main result, we study a bound
called {\it reduced cc-squashed entanglement}. We prove that the bound is convex, and outperforms
both the limitations presented in Refs.~\cite{AL20} and \cite{FBL+21} in certain regime of noise (e.g., see Figure~\ref{DI-QKD-Plots-m}). We then 
show that in the case when testing in DI-QKD protocol is done by estimating the CHSH inequality and the quantum bit error rate (QBER), the cc-squashed entanglement and its reduced version is shown to be a bound, but not studied, in Ref.~\cite{AL20}. We further argue, that the bounds studied in Refs.~\cite{AL20} and \cite{FBL+21} are in fact {\it particular instances of the optimisation} that takes place in computing of the reduced cc-squashed entanglement. 
More precisely,
we show that in the case of a single measurement the FBJ\L{}KA bound is an instance of the optimization of the reduced cc-squashed entanglement with classical Eve. We then extend the reduced cc-squashed entanglement to consider multiple measurements. We prove that the reduced cc-squashed entanglement is a lower bound to the function given in Ref.~\cite{FBL+21}. We note that for extension to multiple measurements, the function introduced is tuned to protocols in which the measurements are announced by Alice and Bob and hence are known to the eavesdropper. We could in principle, on similar grounds, also consider upper bounds for protocols in which the measurements are not known to the eavesdropper.

The upper bounds presented in the above work for limited class of protocols (single pair of inputs, separation between key and testing rounds, or distribution of inputs known in advance).
The only general upper bounds for the DI-QKD achieved by {\it any} protocol
were given in Refs.~\cite{KWW20} and \cite{CFH20}, where generality is due to local operations and public communication (LOPC) mapping of a device to a probability distribution. As our second main result, we provide tighter bounds than in Ref.~\cite{KWW20} and go beyond the results presented in Ref.~\cite{CFH20}. (The latter were restricted only to states with positive partial transposition). The bound is given in terms of reduced relative entropy of entanglement optimized over decompositions into local and non-local part.

So far we have considered a kind of ``static'' approach, upper bounding the DI-QKD rate of a fixed device. We pass now to a more ``dynamic'' one, where we explicitly consider a quantum channel as a part of device.
In Ref.~\cite{CFH20}, the perspective of the provider was developed-- a person who aims at delivering device and checks in advance its limitations. This approach was extended to quantum channels that connect the DI-QKD devices, however the considered examples were restricted to PPT (positive partial transpose) channels~\cite{R99,R01}. As the third main result, we go beyond examples for PPT channels in showing that the device independent private capacity for CHSH protocols can be less than the private capacity of a quantum channel. We show that the well known channels-- {\it the erasure channel} and {\it depolarizing} one, can be simulated in device independent way by {\it a dephasing} channel
corresponding to respective honest realisations. This suggests that the best choice for a provider of the DI-QKD Internet is 
to consider dephasing channel as a mean to distribute CHSH based device independent key. Indeed, the honest parties will obtain only the statistics that can be explained by the use of dephasing channel.

\textit{Note.} Detailed descriptions of the protocols and upper bounds on the DI-QKD rates with several observations are provided in the Supplementary Material. 

{\it Notations}.--- Formally a quantum device is given by
its quantum representation $\Tr M^{x}_a\otimes M^{y}_b \rho$ where $\{M^{x}_a\}_a$ and $\{M^{y}_b\}_b$ are Positive Operator Valued Measures (POVMs) for each input to the device $(x,y)$, and $\rho$ is a bipartite state. We denote such a device as $(\rho,{\cal M})$ where
${\cal M} = \{ M^{x}_a\otimes M^{y}_b\}^{x,y}_{a,b}$. Let $\operatorname{LHV}$ denote the set of states with locally-realistic hidden variable models under given set of measurements.

We denote \textit{iid} device-independent key rate of a quantum device as $K^{iid}_{DI}$. By iid we mean that the devices are independent and identical in each round of the protocol. We then consider various types of protocols for DI-QKD rate. In the first case, single inputs are used for key generation. There are further two variants of such a protocol. The first quantifies the key achieved by protocols in which the honest parties perform test based only on certain parameters of the device. These parameters include the level of  violation of a Bell inequality $\omega(\rho,{\cal M})$  and the rate of the error of the raw key data $P_{err}(\rho,{\cal M})$. We will denote the device independent key rate for such protocols as $K^{iid,(\hat{x},\hat{y})}_{DI,par}(\rho,{\cal M})$ (see Ref.~\cite{AL20} for this approach).
The second, considered in Ref.~\cite{KWW20, FBL+21} is based on the protocols in which the parties perform a full tomography of the device. We denote the DI key rate for such protocols as $K_{DI,dev}^{iid,(\hat{x},\hat{y})}(\rho,{\cal M})$. 

{\it Methods}.--- To study upper bounds on the DI-QKD rates, we begin with introducing an entanglement measure, similar to
the squashed entanglement, which is implicitly used in Ref.~\cite{AL20}. 
\begin{equation}
    E_{sq}^{cc}(\rho_{AB},{\mathrm M}) \coloneqq \inf_{\Lambda_E} I(A:B|E)_{{\mathrm M}\otimes \Lambda_E (\psi^\rho)}
\end{equation}
It is a function of a pair of POVMs 
${\mathrm M}:=M^{\hat{x}}_a\otimes M^{\hat{y}}_a$ and a bipartite quantum state. It computes the infimum over channels acting on the purification $\psi^\rho$ of the state $\rho$, of the {\it conditional mutual information} of resulting extension of $\rho$ measured with ${\mathrm M}$ on system $AB$.
We call
it {\it cc-squashed entanglement } where cc stands for classical-classical registers of the measured system $AB$.

We further consider its reduced versions
({\it reduced cc-squashed entanglement},
where reduction is due to the infimum on the set of allowed
attacking strategies of the eavesdropper while manufacturing the device. As in the case of iid DI key rate there are two versions of reduced cc squashed entanglement:
\begin{align}
    &E_{sq,par}^{cc}(\rho_{AB},{\cal M}(\hat{x},\hat{y})) :=
    \nonumber\\&\inf_{\underset{P_{err}(\sigma,{\cal N})=P_{err}(\rho,{\cal M}(\hat{x},\hat{y}))}{
    \omega(\sigma,{\cal N})=\omega(\rho,{\cal M})}} E_{sq}^{cc}(\sigma_{AB},{\cal M}(\hat{x},\hat{y}))\\
   &E_{sq,dev}^{cc}(\rho_{AB},{\cal M}(\hat{x},\hat{y})) := \inf_{{
    (\sigma,{\cal N})=(\rho,{\cal M})}} E_{sq}^{cc}(\sigma_{AB},{\cal M}(\hat{x},\hat{y}))
\end{align}
In the above definitions, $\left(\rho_{AB}, \mathcal{M}\right)$ corresponds to the device in consideration and $(\hat{x}, \hat{y})$ corresponds to the key generation inputs. In the first definition, we have the infimum over all devices compatible to certain parameters observed in the protocol. In the second definition, we have the infimum over all devices that give the distribution $p(a,b|x,y)= \operatorname{Tr}(\mathcal{M}^x_a\otimes \mathcal{M}^y_b\rho)$. 
It is clear that by definition $E_{sq,par}^{cc}\leq E_{sq,dev}^{cc}$. Using the proof-techniques in Refs.~\cite{CEHHOR,CFH20}, we can see (Theorem~\ref{thm:main-1} and Corollary~\ref{cor:main-1} in the Supplemental Material) that the above quantities upper bound the DI-QKD rate:
\begin{align}
    K_{DI,dev/par}^{iid,(\hat{x},\hat{y})}(\rho,{\cal M}) \leq E_{sq,dev/par}^{cc}(\rho,{\cal M}).
\end{align}

We make a crucial connection that the upper bounds
{\it plotted} in Refs.~\cite{AL20,FBL+21} are upper bounds on $E_{sq,par}^{cc}(\rho_{AB},{\cal M}(\hat{x},\hat{y}))$.
We denote the plotted functions as $I_{AL}(\rho_{AB},{\cal M}(\hat{x},\hat{y}))$ and $I_{FBJL+}(\rho_{AB},{\cal M}(\hat{x},\hat{y}))$ respectively. That means, if $E_{sq,par}^{cc}$ was plotted, it would be lesser than both the bounds $I_{AL}(\rho_{AB},{\cal M}(\hat{x},\hat{y}))$ and $I_{FBJL+}(\rho_{AB},{\cal M}(\hat{x},\hat{y}))$ given in Refs.~\cite{AL20} and \cite{FBL+21}, respectively. For any device $(\rho,{\cal M})$ and input ${\cal M}(\hat{x},\hat{y})$, there is
\begin{align}
    &E_{sq,par}^{cc}(\rho_{AB},{\cal M}(\hat{x},\hat{y}))\leq I_{AL}(\rho_{AB},{\cal M}(\hat{x},\hat{y})),\label{eq:AL}\\
    &E_{sq,dev}^{cc}(\rho_{AB},{\cal M}(\hat{x},\hat{y}))\leq I_{FBJL+}(\rho_{AB},{\cal M}(\hat{x},\hat{y})),\\
    &E_{sq,par}^{cc}(\rho_{AB},{\cal M}(\hat{x},\hat{y}))\leq 
    E_{sq,dev}^{cc}(\rho_{AB},{\cal M}(\hat{x},\hat{y})).
\end{align}

Based on above inequalities and some of the desirable properties like convexity of the cc-squashed entanglement with respect to the states, we arrive at the main Theorem of this Letter. In what follows we narrow considerations to  $({\cal M},(\hat{x},\hat{y}))$ being {\it projective}, as the bound for Werner states presented in Ref.~\cite{FBL+21} applies only to this case.

\begin{theorem}
For a Werner state $\rho_{AB}^W$
and ${\cal M}$ consisting of  projective measurements ${\mathrm P}_a^x\otimes {\mathrm P}_b^y$, and a pair of inputs $(\hat{x},\hat{y})$ used to generate the key, there is
\begin{align}
    &K_{DI,par}^{iid,(\hat{x},\hat{y})}(\rho_{AB}^W,{\cal M})\leq \nonumber\\ &\mbox{Conv}(I_{AL}(\rho_{AB}^W,{\cal M}(\hat{x},\hat{y})),I_{FBJL+}(\rho_{AB}^W,{\cal M}(\hat{x},\hat{y}))),
\end{align}
where $Conv(P_1,P_2)$ is the convex hull of the plots of functions $P_i$, and $K_{DI,par}^{iid,(\hat{x},\hat{y})}(\rho_{AB}^W,{\cal M})$ is defined with respect to $\omega(\rho^W,\mathcal{M})$ and $P_{err}=P(a\neq b|\hat{x}\hat{y})$.
\label{thm:convexification-m}
\end{theorem}

We now extend the definition of $E^{cc}_{sq}$ for multiple measurements. The cc-squashed entanglement of the collection of measurements ${\cal M}$ measured with distribution $p(x,y)$ of the inputs reads:
\begin{equation}
E^{cc}_{sq}(\rho_{AB},{\cal M},p(x,y)):=\sum_{x,y} p(x,y) E_{sq}^{cc}(\rho_{AB},{\mathrm M}_{x,y}).
\end{equation}
and
\begin{equation}
    E^{cc}_{sq,dev}(\rho,\mc{M},p(x,y))\coloneqq \inf_{(\rho,\mc{N})=(\rho,\mc{M})}E^{cc}_{sq}(\sigma,\mc{N},p(x,y)).
\end{equation}

It will appear crucial to notice,
that in DI-QKD it is assumed, that the distribution of inputs $p(x,y)$ is drawn from a private shared randomness held by Alice and Bob, which is independent of the device $(\rho,{\cal M})$.
From here on, we work with the definition of standard protocols. That is, we assume that Alice and Bob make the announcements for the choice of measurements and Eve subsequently learns this measurement~\cite{FBL+21}. To make this assumption explicit
we will consider the following DI-QKD rate:
\begin{align}
    &K_{DI,dev}^{iid,broad}(\rho,{\cal M},p(x,y)):= \nonumber\\
    &\inf_{\epsilon>0}\limsup_n\sup_{{\cal P}\in LOPC}\inf_{{(\sigma,{\cal N})\approx_\epsilon(\rho,{\cal M})}}\nonumber\\
     &\kappa^{\epsilon}_n({\cal P}([\sum_{x,y}p(x,y){\mathrm N}_{xy}\otimes \id_E(\psi^\sigma_{ABE}\otimes |xy\>\<xy|_{E_xE_y})]^{\otimes n})),
\end{align}
where by {\it broad} we mean that $(x,y)$ are broadcasted, and made
explicit by adding systems $E_xE_y$ to 
Eve. Here $\mc{P}$ is a protocol composed of classical LOPC (cLOPC) acting on $n$ identical copies of $(\sigma,\mc{N})$ which, composed with the measurement, results in a quantum LOPC (qLOPC) protocol. Here, $\psi^{\sigma}_{ABE}$ refers to the purification of $\sigma$.

We observe that the $E_{sq}^{cc}(\rho, {\cal M},p(x,y))$ is
an upper bound for distillable key rate
of the state $\sum_{x,y}p(x,y)M_{x,y}\otimes
\id_E\op{\psi^\rho}_{ABE}\otimes\op{x,y}_{E_xE_y}$, where $\psi^\rho_{ABE}$ is purified state of $\rho_{AB}$. This leads us to the following result.
\begin{theorem}
The function $E_{sq,dev}^{cc}(\rho,{\cal M},p(x,y))$ is (i) a convex upper bound on $K_{DI,dev}^{iid,broad}(\rho,{\cal M},p(x,y))$ and (ii) a lower bound to 
the upper bound given in \cite[Eq.~(5)]{FBL+21}.
\end{theorem}
As the second conclusion from the above Theorem there comes the fact that for any family of plots of the upper bound via the average intrinsic information given in Ref.~\cite{FBL+21}, the device independent key rate is below their convex hull.

As we see above $E_{sq,dev}^{cc}(\rho,{\cal M},p(x,y))$ as well as intrinsic non-locality~\cite{KWW20} are based on conditional mutual information where the Eve system is an extension system of underlying strategy. A major difference between the two quantities is that the intrinsic non-locality is a function of the device $\left\{p(a,b|x,y)\right\}$ while $E_{sq,dev}^{cc}(\rho,{\cal M}, p(x,y))$ is a function of the compatible $\rho_{ccq}$ states. The form of the function $E_{sq,dev}^{cc}$ also allows us, following Ref.~\cite{FBL+21}, to designs channels on the eavesdropper system that are dependent on the classical communication between Alice and Bob. The construction of such maps is crucial in obtaining tighter upper bounds. 

We now discuss the second result of our work. We move beyond specific protocols to derive bounds on the general DI-QKD protocols. The maximum DI-QKD rate $K^{iid}_{DI}$ of a device $(\rho,\mc{M})$ is equal to the maximum DI-QKD rate of a device $(\sigma,\mc{N})$ when $(\rho,\mc{M})= (\mc{N},\sigma)$:
\begin{align}
& (\mc{M},\rho)= (\mc{N},\sigma) \implies K^{iid}_{DI,dev}(\rho,\mc{M})= K^{iid}_{DI,dev}(\sigma,\mc{N}).
\end{align}
From the definitions, we also have $K^{iid}_{DI,par}(\rho,\mc{M})\leq K^{iid}_{DI,dev}(\rho,\mc{M})$.

We develop on the results of Ref.~\cite{CFH20} going beyond states that are positive under partial transposition (PPT) \cite{R99,R01}. We arrive at the following upper bound for general protocols:
\begin{theorem}
The maximal DI-QKD rate $K^{iid}_{DI}(\rho,\mc{M})$ of a device $(\rho, \mc{M})$ is upper bounded as
\begin{align}
   K^{iid}_{DI,dev}(\rho,\mc{M})&\leq (1-p)\inf_{ (\sigma^{\textrm{NL}},\mathcal{N})= (\rho^{\textrm{NL}},\mathcal{M})}E_{R}(\sigma^{NL})+\nonumber\\p &\inf_{(\sigma^{\textrm{L}},\mathcal{N})= (\rho^{\textrm{L}},\mathcal{M})}E_{R}(\sigma^{L}),
   \end{align}
   where $E_R(\rho)$ is the relative entropy of entanglement~\cite{Vedral1997} of the bipartite state $\rho$,
   \begin{equation}
      \rho=(1-p)\rho^{\textrm{NL}}+p\rho^{\textrm{L}}\end{equation}
such that $\sigma^{\textrm{L}}, \rho^{\textrm{L}}\in \operatorname{LHV}$ and $\sigma^{\textrm{NL}}, \rho^{\textrm{NL}}\notin \operatorname{LHV}$.
\end{theorem}

A consequence of the above theorem is the following result. The maximal DI-QKD rate $K^{iid}_{DI,dev}(\rho,\mc{M})$ of a device $(\rho,\mc{M})$ under CHSH protocol $\mathcal{P}_{CHSH}$~\cite{E91}, is upper bounded as
\begin{align}\label{eqn-upper-CHSH}
   K^{iid}_{DI,dev}(\rho,\mc{M})&\leq (1-p)\inf_{(\sigma^{\textrm{bl}},\mathcal{N})= (\rho^{\textrm{bnl}},\mathcal{M})}E_{R}(\sigma^{\textrm{bnl}}),
   \end{align}
where $\rho=(1-p)\rho^{\textrm{bnl}}+p\rho^{\textrm{bl}}$ and $\rho^{\textrm{bl}}$ denotes state satisfying CHSH inequality and $\rho^{bnl}$ denotes state violating CHSH inequality~\cite{CHSH}. We plot the upper bounds in Figure~\ref{DI-QKD-Plots1}.
\begin{figure}
    \centering
    \includegraphics[width=\linewidth]{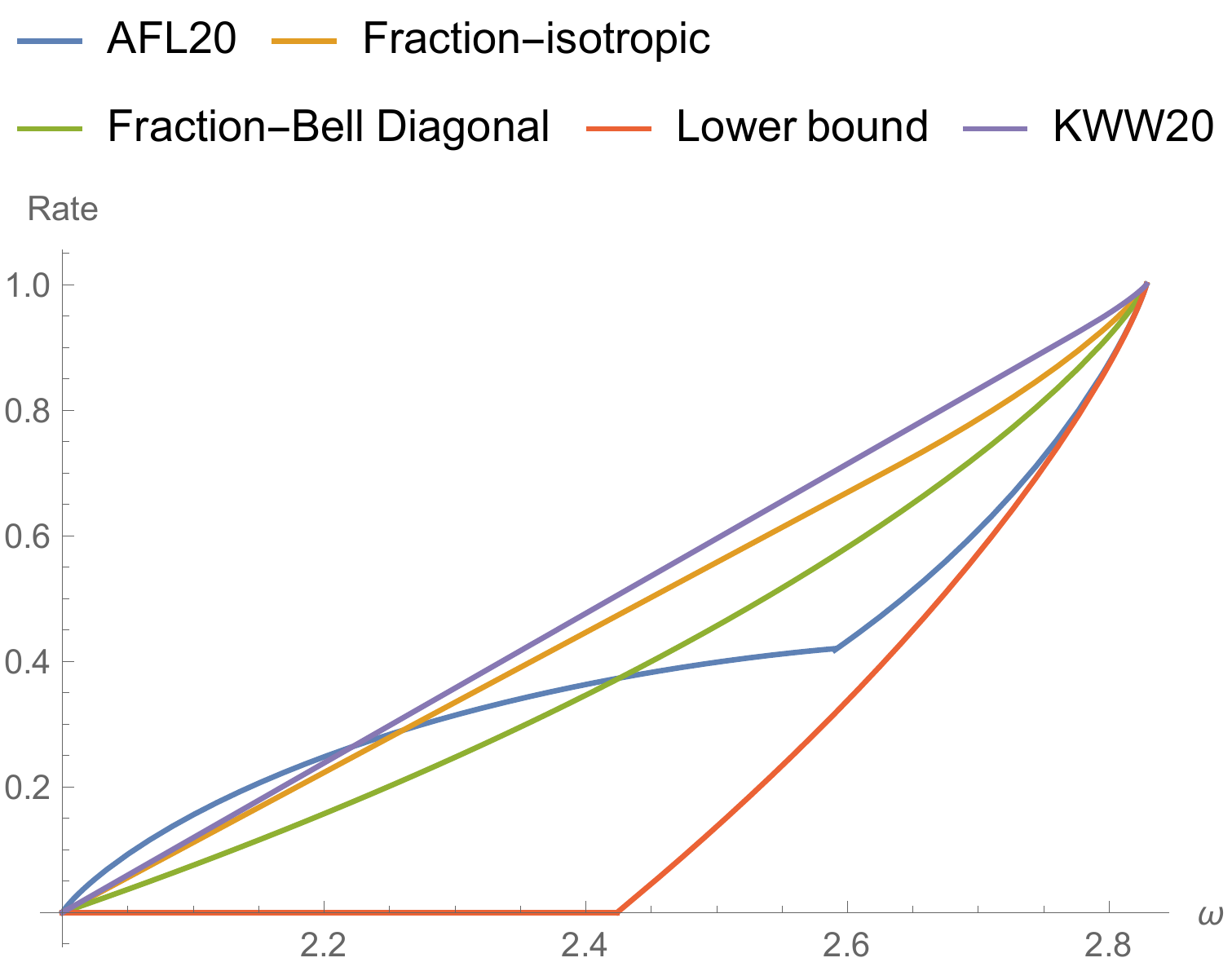}
    \caption{In this plot, we depict the bounds on the amount of DI key that can be obtained from a CHSH-based device. The yellow line and green line corresponds to the upper bounds obtained from \eqref{eqn-upper-CHSH}. The blue line corresponds to the bound obtained in Appendix B of \cite{AL20}. The purple line corresponds to the bound obtained in \cite{KWW20}. The red line corresponds to the lower bounds obtained in \cite{Pironio2009}.}
    \label{DI-QKD-Plots1}
\end{figure}

\begin{figure}
    \centering
    \includegraphics[width=\linewidth]{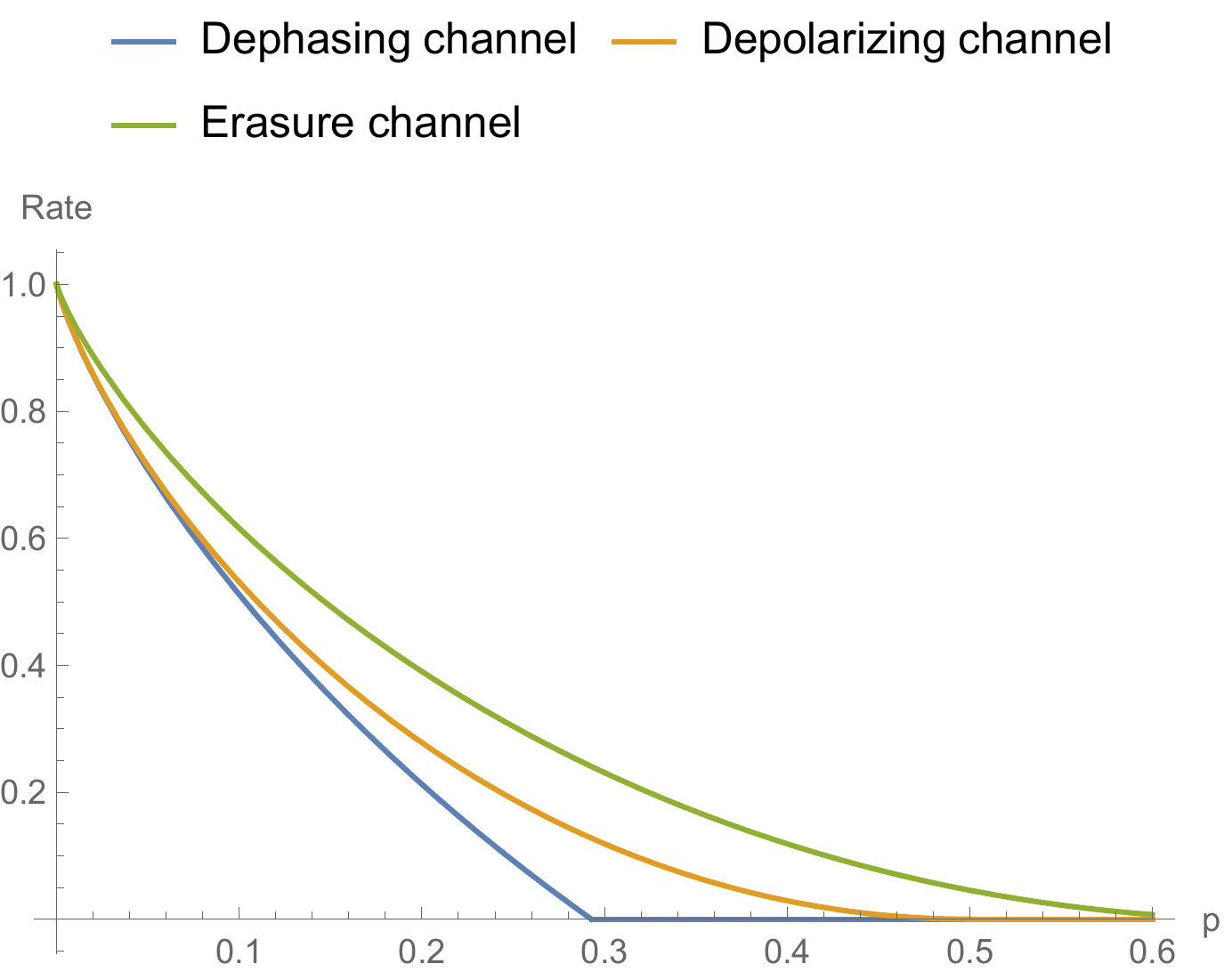}
    \caption{In the above figure, we plot upper bounds on the device-dependent QKD capacities of depolarizing channel (yellow line), dephasing channel (blue line) and erasure channel (green line). We notice that the upper bounds for erasure and dephasing channels are achievable device-dependent QKD rates (capacities). We then notice that for the CHSH protocols, the upper bounds on the DI-QKD capacities of channels is limited by the device-dependent QKD capacity of dephasing channels.}
    \label{fig:channel_capacities-m}
\end{figure}

We now discuss the third main result of our work. It is pertinent from the perspective of a manufacturer to benchmark the maximum DI-QKD rate of a device for their users. Quantum states in the device would be generated by the quantum channels used by the manufacturer to setup the device. We provide upper bounds on the DI-QKD rate of an honest device where the entangled systems are distributed at the measurement ends (apparatus) of Alice and Bob via quantum channels. This is generalization of the setup considered in Ref.~\cite{CFH20}, where rather than a central (relay) point distributing entangled systems to the measurement apparatuses of Alice and Bob, it was assumed that the apparatus at the Alice's end transmits a system from a pair of entangled systems to Bob's end. We discuss the generalized setups and upper bounds on their DI-QKD rates in detail in the Supplementary Material. For now, we provide upper bounds on the device $(\id\otimes\Lambda,\rho,\mc{M})$ with channel setups and protocols as considered in Ref.~\cite{CFH20} albeit considering an arbitrary channel $\Lambda$ from Alice's end to Bob's end. 

The DI-QKD capacity of the device $(\id\otimes\Lambda,\rho,\mc{M})$ under the assumption of its iid uses assisted with $i$-way communication between allies outside the device and $j$-way communication between the input-output rounds within the device, is given by~\cite{CFH20}
\begin{equation}
    \mc{P}^{IDI_j}_{i}(\id\otimes\Lambda,\rho,\mc{M})\coloneqq \inf_{\varepsilon>0}\limsup_{n\to \infty}\mu^{IDI_j,\varepsilon}_{i,n}(\id\otimes\Lambda,\rho,\mc{M}),
\end{equation}
where $\mu^{IDI_j,\varepsilon}_{i,n}(\id\otimes\Lambda,\rho,\mc{M})$ is the maximum key rate optimized over all viable privacy protocols $\hat{\mathcal{P}}$ over the \textit{iid} uses of device, and also includes a minimization over the possible \textit{iid} devices $IDI_j$ that are compatible with the honest device. We have
\begin{align}
   & \mu^{IDI_j,\varepsilon}_{i,n}(\id\otimes\Lambda,\rho,\mc{M}) \nonumber \\
   \coloneqq & \sup_{\hat{P}\in cLOPC_{i}}\inf_{\substack{(\id\otimes\Lambda',\sigma,\mc{N})\in IDI_j\\ (\id\otimes\Lambda',\sigma,\mc{N})\approx_{\varepsilon} (\id\otimes\Lambda,\rho,\mc{M})}}\kappa^{\varepsilon}_n(\hat{P},(\id\otimes\Lambda',\sigma,\mc{N})),
\end{align}
where $\kappa^\varepsilon_n$ is the rate of achieved $\varepsilon$-perfect key and classical labels from local classical operations in $\hat{P}\in cLOPC_{i}$ are possessed by the allies (Alice and Bob).

For the class of channels $\Lambda$ that are simulable via LOCC and the respective Choi states as resource~\cite{Bennett1996,HHH98}, the following upper bounds hold:
     \begin{align} 
  \mc{P}_{i}^{IDI_j}(\id\otimes\Lambda,\rho,\mc{M})
  \leq  \inf_{\substack{(\id\otimes\Lambda',\sigma,\mc{N})\in IDI_j\\ (\id\otimes\Lambda',\sigma,\mc{N})= (\id\otimes\Lambda,\rho,\mc{M})}}  E_R(\Phi^{\Lambda'}),\label{lem:choi-states-channels-m}
 \end{align} 
 where $\Phi^{\Lambda'}\coloneqq \Lambda'(\Phi^+)$ is the Choi state of the channel $\Lambda$, with $\Phi^+_{A'B}\coloneqq \frac{1}{d}\sum_{i,j=0}^{d-1}\ket{i,i}\bra{j,j}_{A'B}$ denoting a maximally entangled state of Schmidt rank $d=\min\{|A'|,|B|\}$.

In Figure~\ref{fig:channel_capacities-m}, we plot the upper bounds on the DI-QKD rates for the devices $(\id\otimes\Lambda,\rho,\mc{M})$ for $\Lambda$ being qubit channels-- depoloarizing $\mathcal{D}^{p}$, dephasing $\mathcal{P}^p$, and erasure $\mathcal{E}^{p}$, where actions of these channels are given as:
    $\mathcal{P}^p(\rho)= (1-p)\rho+p \sigma_Z\rho \sigma_Z$, $\mathcal{D}^{p} (\rho)= (1-p)\rho+p\frac{1}{2}\mathbbm{1}$,  $\mathcal{E}^{p} (\rho)\equiv (1-p)\rho+p\op{e}$, where $\op{e}$ is the erasure symbol, orthonormal to the support of the input state. The relative entropy of entanglement of the Choi states of the erasure and dephasing channels are also the device-dependent QKD capacities of respective channels~\cite{Pirandola2017}. We make a crucial observation that the dephasing channel can simulate the device $(\id\otimes\Lambda,\rho,\mc{M})$ with erasure channel or depolarizing channel in a device-independent way for CHSH protocols. This suggests that the outcomes of the device will have statistics that can be explained by the dephasing channel even when the actual channel present inside device is erasure or depolarizing. Hence, it may be in interest of the manufacturer to use dephasing channel instead of other two channels.
    
{\it Discussion}.---
We have developed tighter bounds on DI-QKD rate
in the case of protocols with single measurement
for generating the raw key. Extending this result
for more measurements (see lower bounds studied in Ref~\cite{schwonnek2020robust}), would be the next important step. We have also developed tighter bounds, based on relative entropy of entanglement, for the general DI-QKD protocols. Developing further on the relative entropic bound, we use it to derive tighter limitations on the DI-QKD rate of bipartite states and setups with quantum channels. Our techniques can be generalized to the multipartite case and will form a future direction.

\begin{acknowledgements}

Part of this work is performed at the Institute for Quantum Computing (IQC), University of 
Waterloo, which is supported by Innovation, Science and Economic Development Canada. EK acknowledges support 
by NSERC under the Discovery Grants Program, Grant 
No. 341495.

This work is part of the ICTQT
IRAP project of FNP. The``International Centre for
Theory of Quantum Technologies" project (contract no.
2018/MAB/5) is carried out within the International Research
Agendas Programme of the Foundation for Polish
Science co-financed by the European Union from the
funds of the Smart Growth Operational Programme, axis
IV: Increasing the research potential (Measure 4.3).
KH thanks Anubhav Chaturvedi for discussion and Tamoghna Das for valuable insight in the topic of upper bounds on device-independent quantum key distribution rates.

SD acknowledges Individual Fellowships at Universit\'{e} libre de Bruxelles; this project receives funding from the European Union's Horizon 2020 research and innovation programme under the Marie Sk\l odowska-Curie grant agreement No.~801505.
\end{acknowledgements}

\section*{The Supplementary Material}
In this Supplementary Material, we elaborate on the results and discussion presented in the Letter. We present detailed proof for the results in the Letter along with some additional relevant results and observations.

\subsection{Bounds on device-independent key distillation rate of states}\label{sec:di-qkd-state}
Let $\mc{D}(\mc{H}_{AB})$ denote the set of states defined on the $\mc{H}_{AB}\coloneqq \mc{H}_{A}\otimes\mc{H}_B$, where $\mc{H}_{A}$ and $\mc{H}_B$ are the separable Hilbert spaces associated with the quantum systems $A$ and $B$, respectively. Let $\mathbbm{1}_A$ denote the identity operator on $\mathcal{H}_{A}$ and $|A|$ denote the dimension of $\mc{H}_A$ ($\dim(\mc{H}_A)$). Let $\Phi^+_{AB}$ denote a maximally entangled state,
\begin{equation}
    \Phi^+_{AB}\coloneqq \frac{1}{d}\sum_{i,j=0}^{d-1}\ket{i,i}\bra{j,j}_{AB}
\end{equation}
for $d=\min\{|A|,|B|\}$ and an orthonormal basis $\{\ket{i}\}_i$. Let $\vartheta$ be the partial transposition map with respect to a fixed basis, i.e., $\vartheta(\rho_{AB})=\rho_{AB}^{\Gamma_{B}}$.

Consider a setup, wherein Alice and Bob, two spatially separated parties, have to extract a secret key. We assume that in this setup, the devices are untrusted. That is, Alice and Bob do not trust the quantum states, nor do they trust their measurement devices. The untrusted measurement of the device is given by  $\mc{M}\equiv \{M^x_a\otimes M^y_b\}_{a,b|x,y}$, where $a \in \mathcal{A}$, $b\in \mathcal{B}$, and $\mathcal{A},\mathcal{B}$ denote the finite set of measurements outcomes. The measurement outcomes, i.e., ouputs of the device, are secure from adversary and assumed to be in the possession of the receiver, Alice or Bob. Also, $x \in \mathcal{X}$, $y\in \mathcal{Y}$, where $\mathcal{X},\mathcal{Y}$ denote the finite set of measurements choices. The joint probability distribution is given as $p(a,b|x,y)=\Tr[M^x_a\otimes M^y_b \rho]$ for measurement $\mc{M}$ on bipartite state $\rho_{AB}$ defined on the separable Hilbert space $\mc{H}_A\otimes\mc{H}_B$. The tuple $\left\{\rho,\left\{M^x_a\right\}_x,\left\{M^y_b\right\}_y\right\}$ is called the quantum strategy of the distribution. The quantum systems $A,B$ can be finite- or infinite-dimensional. The number of inputs $x,y$ and corresponding outputs $a,b$ of local measurements by Alice and Bob are arbitrary in general. 

Let $\omega(\rho, \mc{M})$ denotes the violation of the given Bell inequality $\mc{B}$ by state $\rho_{AB}$ when the measurement settings are given by $\mc{M}$. Let $P_{err}(\rho,\mc{M})$ denote the expected qubit error rate (QBER). Both the Bell violation, as well as the QBER are a function of the probability distribution of the box. If under local measurements $\mc{M}$, a state $\rho$ exhibits a locally-realistic hidden variable model then we write $\rho\in\operatorname{LHV}(\mc{M})$. If a state $\sigma$ satisfies Bell inequality $\mc{B}$ under local measurements $\mc{M}$ then we write $\sigma\in\operatorname{L}(\mc{B},\mc{M})$. If a Bell inequality is a facet of local polytope then the set $\operatorname{LHV}(\mc{M})$ of states with locally-realistic hidden variable model and the set $\operatorname{L}(\mc{B},\mc{M})$ of states that satisfy given Bell inequality are equal. An example of such a Bell inequality is CHSH inequality for which $\operatorname{LHV}(\mc{M})=\operatorname{L}(\mc{B}=CHSH, \mc{M})$. We will use $\SEP(A\!:\!B)$ to denote the set of separable states defined on $\mc{H}_{A}\otimes\mc{H}_B$.

If $\{p(a,b\vert x,y)\}_{a,b|x,y}$ obtained from $(\rho,\mc{M})$ and $(\sigma,\mc{N})$ are the same, we write $(\sigma,\mc{N})= (\rho,\mc{M})$. In most DI-QKD protocols, instead of using the statistics of the \textit{full} correlation, we use the Bell violation and the QBER to test the level of security of the observed statistics. In this way, the protocols coarse grain the statistics and we only use partial information of the full statistics to extract the device-independent key. 
In this context, the notation $(\sigma,\mc{N})= (\rho,\mc{M})$ also implies that $\omega(\sigma,\mc{N})= \omega(\rho,\mc{M})$ and $P_{err}(\sigma,\mc{N})=P_{err}(\rho,\mc{M})$. When conditional probabilities associated with $(\rho,\mc{M})$ and $(\sigma,\mc{N})$ are $\varepsilon$-close to each other, then we write $(\rho,\mc{M})\approx_{\varepsilon}(\sigma,\mc{N})$. For our purpose, it suffices to consider the distance
\begin{equation}
    d(p,p')=\sup_{x,y}\norm{p(\cdot\vert x,y)-p'(\cdot \vert x,y)}_1\leq \varepsilon.
\end{equation}

The device-independent (DI) distillable key rate of a device is informally defined as the supremum over the finite key rates $\kappa$ achievable by the best protocol on any device compatible with $(\rho,\mc{M})$, within an appropriate asymptotic blocklength limit and security parameter. Another approach taken is to minimize the key rate over statistics compatible with Bell parameter and a quantum bit error rate (QBER) (e.g.,~\cite{AL20}). For our purpose, we constrain ourselves to the situation when the compatible devices are supposedly \textit{iid} (independent and identically distributed).

Consider the following relations:
\begin{align}
    (\rho,\mc{M}) &\approx_\varepsilon  (\sigma,\mc{N}) \label{eq:k-1}\\
    \omega(\rho,\mc{M})& \approx_{\varepsilon} \omega(\sigma,\mc{N})\label{eq:k-2}\\
    P_{err}(\rho,\mc{M})& \approx_{\varepsilon} P_{err}(\sigma,\mc{M})\label{eq:k-3},
\end{align}
where ~\eqref{eq:k-1} implies ~\eqref{eq:k-2} and ~\eqref{eq:k-3}.
Formally, the definition of device-independent distillable key rate is given as
\begin{definition}[\cite{CFH20}]
The maximum device-independent key rate of a device $(\rho,\mc{M})$ with $iid$ behavior is defined as
\begin{equation}
    K^{iid}_{DI}(\rho,\mc{M})\coloneqq \inf_{\varepsilon>0}\limsup_{n\to \infty} \sup_{\hat{\mc{P}}} \inf_{\eqref{eq:k-1}} \kappa^\varepsilon_{n}(\hat{\mc{P}}(\sigma,\mc{N})^{\otimes n}),
\end{equation}
where $\kappa_n^\varepsilon$ is the key rate achieved for any security parameter $\varepsilon$, blocklength or number of copies $n$, and measurements $\mc{N}$.

Here, $\hat{\mc{P}}$ is a protocol composed of classical local operations and public (classical) communication (cLOPC) acting on $n$ identical copies of $(\sigma,\mc{N})$ which, composed with the measurement, results in a quantum local operations and public (classical) communication (qLOPC) protocol.
\end{definition}

The following Lemma follows from the definition of $K^{iid}_{DI}$ (see ~\cite{CFH20} for the proof argument made for $\leq$):
\begin{lemma}\label{lem:di-state}
The maximum device-independent key rate $K^{iid}_{DI}$ of a device $(\rho_{AB},\mc{M})$ is equal to the maximum device-independent key rate of a device $(\sigma,\mc{N})$ when $(\rho,\mc{M})= (\sigma,\mc{N})$:
\begin{align}
& (\rho,\mc{M})=(\sigma,\mc{N}) \implies K^{iid}_{DI}(\rho,\mc{M})= K^{iid}_{DI}(\sigma,\mc{N}).\label{eq:p-d-i}
\end{align}
\end{lemma}

\begin{definition}\label{def:di-state}
The maximum device-independent key distillation rate $K^{iid}_{DI}$ of a bipartite state $\rho_{AB}$ is given by
\begin{equation}
    K^{iid}_{DI}(\rho)=\sup_{\mc{M}}K^{iid}_{DI}(\rho,\mc{M}).
\end{equation}
\end{definition}

\begin{observation}
We note that there may exist states $\rho$ for which $K^{iid}_{DI}(\rho)=0$ but $K^{iid}_{DI}(\rho^{\otimes k})> 0$ for some $k\in\mathbbm{N}$.
\end{observation}

A bipartite state $\rho$ that is positive under partial transposition (PPT), i.e., $\rho^{\operatorname{\Gamma}}\geq 0$, is called a PPT state. Similarly, a point-to-point channel $\Lambda$ is called PPT channel if $\Lambda\circ\vartheta$ is also a quantum channel~\cite{R99,R01}, where $\vartheta$ is a partial transposition map, i.e., $\vartheta(\rho)=\rho^{\Gamma}$. There exists bipartite entangled states which are PPT~\cite{P96,S00,HHH98}. However, all PPT states are useless for the task of entanglement distillation via LOCC even if they are entangled~\cite{HHH98,R99, R01}.

A direct consequence of the Lemma~\ref{lem:di-state} is the following corollary (see~\cite{CFH20} for the proof argument made for $\leq$).
\begin{corollary}
For any bipartite state $\rho$ that is PPT, we have $K^{iid}_{DI}(\rho)= K^{iid}_{DI}(\rho^{\operatorname{\Gamma}})$.
\end{corollary}

As discussed above, a large class of device-independent quantum key distribution protocols, rely on the Bell violation and the QBER of the device $p(a,b|x,y)$. For such protocols, we can define the device-independent key distillation protocol as:

\begin{definition}[cf.~\cite{AL20}]
The maximal device-independent key rate of a device $(\rho,\mc{M})$ with $iid$ behavior, Bell violation $\omega(\rho,\mc{M})$ and QBER $P_{err}(\rho,\mc{M})$, is defined as
\begin{align}
    & K^{iid}_{DI}(\rho,\mc{M},\omega, P_{err})\nonumber \\
    &\quad\coloneqq \inf_{\varepsilon>0}\limsup_{n\to \infty} \sup_{\hat{\mc{P}}} \inf_{\eqref{eq:k-2},\eqref{eq:k-3}} \kappa^\varepsilon_{n}(\hat{\mc{P}}(\mc{N},\sigma)^{\otimes n}).
\end{align}
\end{definition}

The set of protocols are restricted for $K^{iid}_{DI}(\rho,\mc{M})$. We then obtain, from the definitions that, $K^{iid}_{DI}(\rho,\mc{M},\omega, P_{err})\leq K^{iid}_{DI}(\rho,\mc{M})$.

The maximal device-independent key distillation rate $K_{DI}(\rho, \mathcal{M})$ for the device $(\rho,\mc{M})$ is upper bounded by the maximal device-dependent key distillation rate $K_{DD}(\sigma)$ for all $(\sigma,\mc{N})$ such that $(\sigma,\mc{N})= (\rho,\mc{M})$ (see~\cite{CFH20}), i.e.,
\begin{equation}
    K_{DI}(\rho, \mathcal{M})\leq \inf_{(\sigma,\mc{N})= (\rho,\mc{M})} K_{DD}(\sigma)\label{eq:di-bound10}
\end{equation}
The device-dependent key distillation rate $K_{DD}(\rho)$ is the maximum secret key (against quantum eavesdropper) that can be distilled between two parties using local operations and classical communication (LOCC) \cite{RennerThesis,HHHO09}. 

\begin{observation}
For entanglement measures $Ent$ which upper bounds the maximum device-dependent key distillation rate, i.e., $K_{DD}(\rho)\leq Ent (\rho)$ for a density operator $\rho$, we have,
\begin{align}\label{eq:di-e-1}
     K_{DI}(\rho, \mathcal{M}) &\leq \inf_{(\sigma,\mc{N})= (\rho,\mc{M})} K_{DD}(\sigma) \\
&     \leq  \inf_{(\sigma,\mc{N})= (\rho,\mc{M})} Ent (\sigma).
\label{eq:FCHbound}
\end{align}
Some well-known entanglement measures that upper bound $K_{DD}(\rho)$ are the relative entropy of entanglement $E_{R}(\rho)$, regularized relative entropy of entanglement $E^{\infty}_{R}(\rho)$, squashed entanglement $E_{sq}(\rho)$ (see~\cite{HHHO09,CW04, C06}).
\end{observation}

We develop on the results of Ref.~\cite{CFH20} going beyond states that are positive under partial transposition (PPT). We arrive at the following main result which holds for general protocols, rather than  specific ones considered in Ref.~\cite{AL20,FBL+21}. 

We now state the upper bound for the DI-key distillation in form of the following theorem:
\begin{theorem}
The maximal device-independent key rate $K^{iid}_{DI}(\rho,\mc{M})$ of a device $(\rho,\mc{M})$ is upper bounded as
\begin{align}
   K^{iid}_{DI}(\rho,\mc{M})&\leq (1-p)\inf_{ (\sigma^{\textrm{NL}},\mathcal{N})= (\rho^{\textrm{NL}},\mathcal{M})}E_{R}(\sigma^{NL})+\nonumber\\p &\inf_{(\sigma^{\textrm{L}},\mathcal{N})= (\rho^{\textrm{L}},\mathcal{M})}E_{R}(\sigma^{L}),
   \end{align}
   where 
   \begin{equation}
      \rho=(1-p)\rho^{\textrm{NL}}+p\rho^{\textrm{L}}  
   \end{equation}
such that $\sigma^{\textrm{L}}, \rho^{\textrm{L}}\in \operatorname{LHV}$ and $\sigma^{\textrm{NL}}, \rho^{\textrm{NL}}\notin \operatorname{LHV}$.
\end{theorem}
\begin{proof}
The main idea behind the proof is to construct a new state with flagged local and non local parts, followed by construction of flagged POVM elements. The flagged strategy reproduces the exact statistics as the original strategy. With this construction, we can use the decomposition of relative entropy for flagged state as the sum of the relative entropies for the constituent states. 

Let us begin with the device $\left\{\rho,\mathcal{M}\right\}$ such that 
\begin{equation}
    \rho_{AB}=(1-p)\rho^{\textrm{NL}}+p\rho^{\textrm{L}}
\end{equation}
We have from \eqref{eq:di-bound10} that
\begin{equation}
    K_{DI}(\rho,\mathcal{M})\leq \inf_{(\sigma,\mathcal{N})=(\rho,\mathcal{M})}K_{DD}(\sigma)
\end{equation}
Let us now construct a strategy $\left\{\sigma^{NL},N^x_a\otimes N^y_b\right\}$ such that $(\sigma^{NL},N^x_a\otimes N^y_b)= (\rho^{NL},M^x_a\otimes M^y_b)$. We also construct another strategy $\left(\sigma^{L},\Lambda^x_a\otimes\Lambda^y_b\right)$ such that $(\sigma^{L},\Lambda^x_a \otimes \Lambda^y_b)= (\rho^{L},M^x_a\otimes M^y_b)$. Combining the above deductions, we can define a strategy 
\begin{align}
    \sigma_{ABR_1R_2} &=(1-p)\sigma^{NL}_{AB}\otimes \op{0}_{R_1}\otimes\op{0}_{R_2}\nonumber\\ +&p\sigma_{L}\otimes \op{1}_{R_1}\otimes\op{1}_{R_2}\\
    \tilde{\Lambda}^x_a&=N^x_a\otimes \op{0}_{R_1}+\Lambda^x_a\otimes \op{1}_{R_1}\\
    \tilde{\Lambda}^y_b&=N^y_b\otimes \op{0}_{R_2}+\Lambda^y_b\otimes \op{1}_{R_2}.
\end{align}
We then see $(\rho,\mathcal{N})=(\sigma,\tilde{\Lambda}^x_a \otimes\tilde{\Lambda}^y_b)$. We then obtain
\begin{align}
    K_{DI}(\rho,\mathcal{M})&\leq K_{DD}(\sigma)\\&\leq E_{R}(\sigma)\\
    &=(1-p)E_R(\sigma^{NL})+pE_R(\sigma^L)
\end{align}
Since the strategies $\left(\sigma^L,{\Lambda}^x_a\otimes\Lambda^y_b\right)$ and $\left(\sigma^{NL},{N}^x_a\otimes N^y_b\right)$ are arbitrary strategies, we obtain
\begin{align}
    K_{DI}(\rho,\mathcal{M})& \leq (1-p)\inf_{\left(\rho^{\textrm{NL}},\mathcal{M}\right)= \left(\sigma^{\textrm{NL}},\mathcal{N}\right)}E_{R}(\sigma^{NL})+\nonumber\\p &\inf_{\left(\rho^{\textrm{L}},\mathcal{M}\right)= \left(\sigma^{\textrm{L}},\mathcal{N}\right)}E_{R}(\sigma^{L})
\end{align}
\end{proof}

We now define a CHSH protocol considered in Ref.~\cite{CHSHprotocol2006}. In this protocol, Alice's device has three inputs, i.e. $x\in\left\{0,1,2\right\}$ and Bob's device has two input, i.e., $y\in\left\{0,1\right\}$. Alice and Bob's output are binary, i.e., $a,b \in\left\{0,1\right\}$. This device is then defined by the distribution $\left\{p(a,b|x,y)\right\}$. The protocol uses a coarse graining of the distribution. That is, for each distribution, we define the CHSH violation $\beta$ and the QBER $p(a\neq b|x=0,y=1)$ as $q$. For such protocols $\mathcal{P}_{CHSH}$, the relevant statistics of the device are $\beta$ and $q$. We have the following:
\begin{corollary}\label{cor:CHSH}
    The maximal device-independent key rate $K^{iid}_{DI}(\rho,\mc{M})$ of a device $(\rho,\mc{M})$ under CHSH protocol $\mathcal{P}_{CHSH}$, is upper bounded as
\begin{align}
   K^{iid}_{DI}(\rho,\mc{M},\mathcal{P}_{CHSH})&\leq (1-p)\inf_{(\sigma^{\textrm{NL}},\mathcal{N})= (\rho^{\textrm{NL}},\mathcal{M})}E_{R}(\sigma^{\textrm{NL}}),
   \end{align}
where 
   \begin{equation}
      \rho=(1-p)\rho^{\textrm{NL}}+p\rho^{\textrm{L}},  
   \end{equation}
and $\rho^{\textrm{L}}\in \operatorname{L}$ and $\sigma^{NL},\rho^{NL}\notin \operatorname{L}$ (for respective local measurements and CHSH inequality).
\end{corollary}
\begin{proof}

To see the proof, we construct a strategy $\left\{\rho_{\textrm{SEP}}, N^x_a\otimes N^y_b\right\}$ that reproduces $\beta= 2$ and any value of $q$. For this, choose $\rho_{\textrm{SEP}}=\op{00}$, and choose $N^{1}_a=N^2_a=N^1_b=\sigma_z$, and $N^2_b=\sigma_x$. It is easy to check that this gives us a CHSH violation of 2. For QBER, the appropriate measurements are $N^0_a,N^1_b$. We see that any value $q$ of QBER can be obtained by choosing $N^0_a$ as $\sigma_x$ with probability $2q$ and $\sigma_z$ with probability $1-2q$. Since the relative entropy of entanglement of $\op{00}$ is 0, we have the proof. 
\end{proof}

\begin{remark}
All bound entangled states satisfy CHSH inequality~\cite{Masanes2006} and hence have zero device-independent key rate for CHSH-based protocols. It is interesting to note that there exists bound entangled states from which private key can be distilled~\cite{HHHO09}, however even these states are useless for device-independent secret key distillation for CHSH-based protocols. That the same can be extended to all Bell inequalities is a matter of recently posed conjecture~\cite{AL20}-- revised Peres conjecture.
\end{remark}

\subsection{Numerics for CHSH protocols}
We now plot the results of Corollary
\ref{cor:CHSH}. We first consider the device to be an honest device with the underlying state being an isotropic state. The honest measurements are $M^0_a=\sigma_z, M^1_a=\frac{\sigma_z+\sigma_x}{\sqrt{2}}, M^2_a= \frac{\sigma_z-\sigma_x}{\sqrt{2}}, M^0_b=\sigma_z, M^1_b=\sigma_x$, where $\sigma_x$ and $\sigma_z$ are Pauli-$x$ and Pauli-$z$ operators, respectively. A device having the above honest realization is called a CHSH-based device. Let us now consider the observed CHSH violation to be $\omega$. Then, we can construct an isotropic state  as
\begin{equation}
   \rho^{\nu}= (1-\nu)\op{\Phi^+}+(\nu/4)\mathbbm{1},
\end{equation}
where the parameter $\nu$ is related to the CHSH violation as, $\omega = 2\sqrt{2}(1-\nu)$. We then have
\begin{equation}
    \rho^{\omega}=\frac{\omega}{2\sqrt{2}}\op{\Phi^+{}}+\frac{2\sqrt{2}-\omega}{8\sqrt{2}}I.
\end{equation}
The relative entropy of entanglement for the isotropic state is given by 
\begin{align}
    E_R(\rho^\omega) &= \lambda \log_2 \lambda + (1-\lambda)\log_2 (1-\lambda)+1,\\
    \lambda &= \frac{\omega}{2\sqrt{2}}+\frac{2\sqrt{2}-\omega}{8\sqrt{2}}\\
    &=\frac{3\omega}{8\sqrt{2}}+\frac{1}{4}.
    \end{align}
Now, we express $\rho^\omega$ as
\begin{equation}
    \rho^{\omega} = p\rho^{\omega_1}+(1-p)\rho^{\omega_2},
\end{equation}
where $2\sqrt{2}\geq\omega_1>2$ and $2\geq\omega_2\geq0$ and $\omega=p\omega_1+(1-p)\omega_2$.

We then have the following optimization:
\begin{align}
    K^{iid}_{DI}(\rho^\omega,\mc{M},\mathcal{P}_{CHSH})&=\ \min_{p,\ \omega_1}\ p\ E_R(\rho^{\omega_1}),\label{eq:opt2}\\
    &\omega=p\ \omega_1+(1-p)\ \omega_2,\\
    &\quad0\leq p\leq 1,\\
    &\quad 2\leq\omega_1<2\sqrt{2},\\
    &\quad 0\leq\omega_2\leq2. \label{eq:opt1}
\end{align}
By performing this optimization, we can obtain the bound given in the Figure \ref{DI-QKD-Plots}.

Another approach that we can take is to consider the quantum strategy taken in Ref.~\cite{Pironio2009}. For this strategy the attacking quantum state is 
\begin{align}
    \rho &= \frac{1+C}{2}\op{\Phi^{+}}+ \frac{1-C}{2}\op{\Phi^{-}},\label{eq:opt-state}\\
    \ket{\Phi^{+}}_{AB} &= \frac{1}{\sqrt{2}}\left(\ket{00}+\ket{11}\right),\\
    \ket{\Phi^{-}}_{AB} &= \frac{1}{\sqrt{2}}\left(\ket{00}-\ket{11}\right),\\
    C&=\sqrt{\left(\frac{\omega}{2}\right)^2-1}.
\end{align}
With this strategy we obtain a tighter upper bound as plotted in the Figure \ref{DI-QKD-Plots}. This strategy is particularly interesting as the attacking state does not allow for a decomposition between a local and a non-local part. The fractional bound reduces to relative entropy of entanglement of the state. 
\begin{figure}
    \centering
    \includegraphics[width=\linewidth]{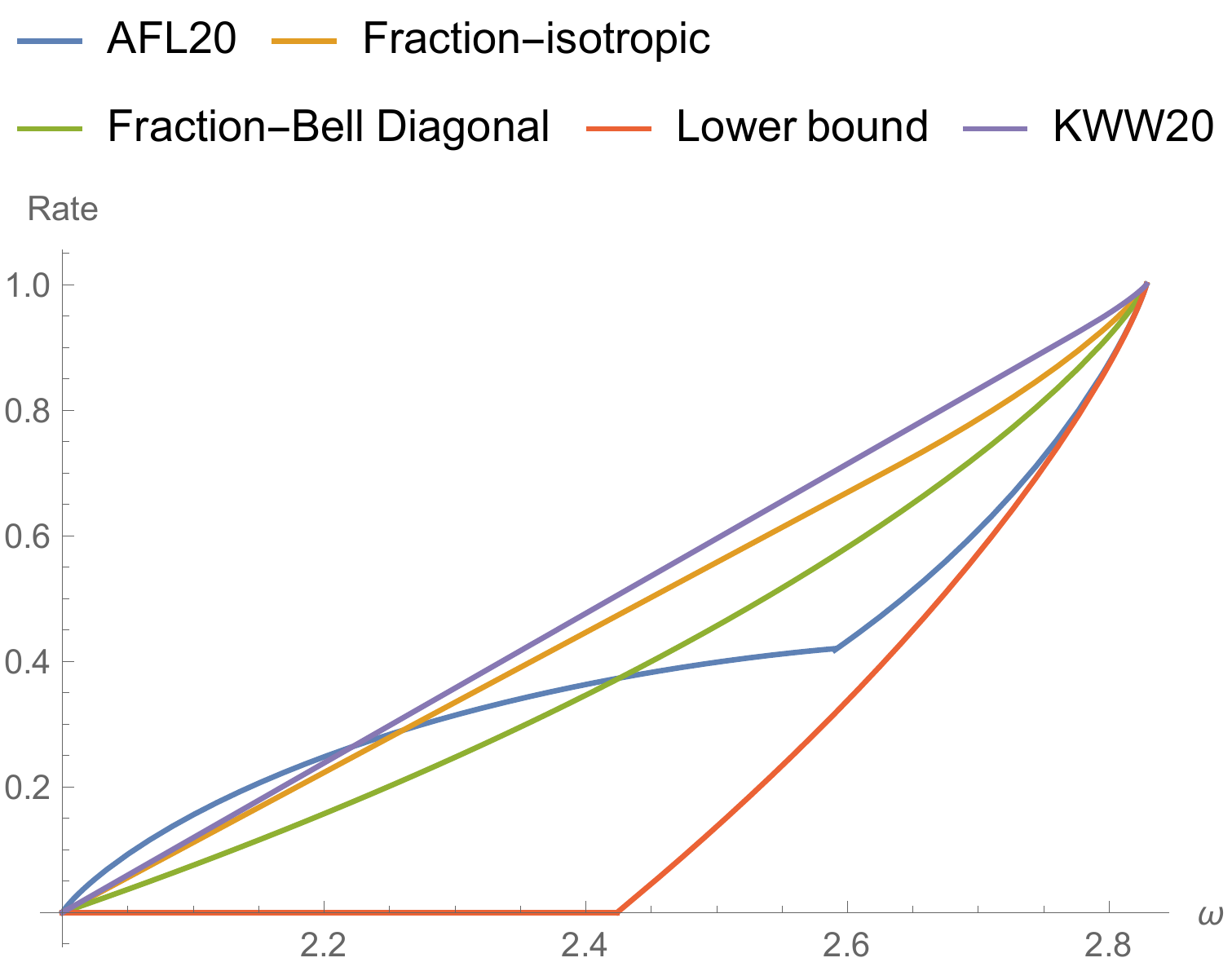}
    \caption{In this plot, we depict the bounds on the amount of DI key that can be obtained from a CHSH-based device. The yellow line corresponds to the upper bounds obtained from optimization given in \eqref{eq:opt2}-\eqref{eq:opt1} and green line corresponds to the upper bound obtained from choosing the attack given in \eqref{eq:opt-state}. The blue line corresponds to the bound obtained in Appendix B of \cite{AL20}. The purple line corresponds to the bound obtained in \cite{KWW20}. The red line corresponds to the lower bounds obtained in \cite{Pironio2009}.}
    \label{DI-QKD-Plots}
\end{figure}

\subsection{Bounds on device-independent key distillation rate through channels}

A simple realistic model of a physical box depicting device-independent secret key generator (assumed to be an honest device from perspective of manufacturer) between allies, is describable by a tuple $(\widetilde{\Omega}, \rho,\mc{M})$. Tuple for device constitutes of measurement setting $\{M^x_a\otimes M^y_b\}$, a source state $\rho_{A'B'}$, and a bipartite quantum distribution channel $\widetilde{\Omega}_{A'B'\to AB}$, where a relay station inputs bipartite quantum state and each output of the channel $\Omega$ is transmitted to designated receiver. The dimensions of the quantum systems $A',B',A,B$ can be arbitrary in general as device can use an arbitrary bipartite quantum channel $\widetilde{\Omega}$. Quantum states from the source undergo quantum dynamical evolution (quantum channels) before they are measured to yield outputs at the ends of Alice and Bob, who are designated parties/allies. Quantum channels can represent noisy transmission via optical fibers, space, etc. or local time-evolution. In general, the bipartite distribution channel $\widetilde{\Omega}_{A'B'\to AB}$ is of the form $\Lambda^{1}_{A^{''}\to A}\otimes\Lambda^{2}_{B''\to B}\circ\Omega_{A'B'\to A''B''}$, where $\Omega_{A'B'\to A''B''}$ allows for the joint operation on the bipartite source state and $\Lambda^{1}_{A^{''}\to A}, \Lambda^{2}_{B^{''}\to B}$ transmit $A'', B''$ to the ends $A,B$ where local measurements (temporally sync between Alice and Bob) take place to yield classical outputs $a,b$ to Alice and Bob, respectively. In general, adversarial manufacturer can design the device such that it can perform any physical actions between the rounds and is only required to provide two pairs of classical input-outputs, a pair to each designated party, while adversary is limited only by the laws of quantum mechanics (which includes no-signaling).

The probability distribution $p_{(\widetilde{\Omega},\rho,\mc{M})}(a,b|x,y)$ associated with an honest device $(\widetilde{\Omega},\rho,\mc{M})$ is given by 
\begin{equation}
    p(a,b|x,y)=\Tr[M^x_a\otimes M^y_b (\widetilde{\Omega}(\rho))].
\end{equation}
An honest device $(\widetilde{\Omega},\rho,\mc{M})$ constituting channels has same characterization/conditions as an honest device $(\widetilde{\Omega}(\rho),\mc{M})$, which is ideal situation where time-evolution or noisy transmission of the source state $\widetilde{\Omega}(\rho)$ before measurement at the ends of Alice and Bob is not considered (see Section~\ref{sec:di-qkd-state}), i.e., in principle $(\widetilde{\Omega},\rho,\mc{M})=(\widetilde{\Omega}(\rho),\mc{M})$. If $\{p(a,b|x,y)\}$ obtained from the devices $(\widetilde{\Omega},\rho,\mc{M})$ and $(\widetilde{\Omega}',\rho',\mc{M}')$ are the same, then we write $(\widetilde{\Omega},\rho,\mc{M}) = (\widetilde{\Omega}',\rho',\mc{M})$. The Hilbert space dimensions associated with systems involved in the device $(\widetilde{\Omega}',\rho',\mc{M}')$ need not be the same as their counterpart systems associated with the honest device $(\widetilde{\Omega},\rho,\mc{M})$. We have
\begin{align}
    (\widetilde{\Omega},\rho,\mc{M}) = (\widetilde{\Omega}',\rho',\mc{M}') 
 & \iff  p_{(\widetilde{\Omega},\rho,\mc{M})}  = p_{(\widetilde{\Omega}',\rho',\mc{M}')},\\
 (\widetilde{\Omega},\rho,\mc{M}) \approx_{\varepsilon} (\widetilde{\Omega}',\rho',\mc{M}')
 & \iff  p_{(\widetilde{\Omega},\rho,\mc{M})} \approx_{\varepsilon} p_{(\widetilde{\Omega}',\rho',\mc{M}')}.
\end{align}

The device-independent secret key agreement protocols can allow $i$-way communication for $i\in\{0,1,2\}$ depending on whether two allies, Alice and Bob, are allowed to perform $i$-way classical communication outside the devices~\cite{CFH20}. This classical communication includes error correction and parameter estimation rounds. A $2$-way LOPC ($LOPC_2$) is a general LOPC channel where both Alice and Bob can send classical communication to each other over authenticated public channel, $1$-way LOPC ($LOPC_1$) is a restricted class of $LOPC$ where only one party is allowed to transmit classical communication to the other (while the other remains barred from sending classical communication), and $0$-way LOPC ($LOPC_0=LO$) is very restricted class of LOPC where both parties can only perform local operations and barred from any classical communication. Therefore $LO\subset LOPC_1\subset LOPC_2$. We note that in practice, there is also need for (classical) communication to agree upon the protocol, and for other purposes like verification/testing. 

Apart from classical communication between Alice and Bob during the key distillation protocol,
there also exists possibility of the classical communication in the device. This classical communication can be based on the inputs from the previous the rounds, that can be used by the device to prepare the source state to be measured in the coming round~\cite{CFH20}. $DI_j$ denotes the devices where the channel $\widetilde{\Omega}$ is \textit{iid}, memory is allowed, and use $j$-way (classical) communication between the input-output rounds for $j\in\{0,1,2\}$. This $j$-way communication can take place either before the inputs are given or after the outputs are obtained. The $DI_j$ devices can share memory locally at Alice and Bob across each round enabling the capability of adversary~\cite{CFH20}.

If an honest device $(\widetilde{\Omega},\rho,\mc{M})$ constituting channels is being used just for a single round (where the bipartite distribution channel $\widetilde{\Omega}$ is called just once) then it is same as an honest device $(\widetilde{\Omega}(\rho),\mc{M})$, which is ideal situation where time-evolution or noisy transmission of the source state $\widetilde{\Omega}(\rho)$ before measurement at the ends of Alice and Bob is not considered (see Section~\ref{sec:di-qkd-state}). That is $(\widetilde{\Omega},\rho,\mc{M})=(\widetilde{\Omega}(\rho),\mc{M})$ for a single round where the device uses channel just once.

\begin{definition}
The device-independent secret key agreement (or private) capacity of the device $(\widetilde{\Omega},\rho,\mc{M})$ assisted with $i$-way communication between allies outside the device and $j$-way communication between the input-output rounds within the device, is given by
\begin{equation}
    \mc{P}^{DI_j}_{i}(\widetilde{\Omega},\rho,\mc{M})\coloneqq \inf_{\varepsilon>0}\limsup_{n\to \infty}\mu^{DI_j,\varepsilon}_{i,n}(\widetilde{\Omega},\rho,\mc{M}),
\end{equation}
where $\mu^{DI_j,\varepsilon}_{i,n}(\widetilde{\Omega},\rho,\mc{M})$ is the maximum key rate optimized over all viable privacy protocols, while also including a minimization over the possible devices $DI_j$ that are compatible with the honest device.
\end{definition}

While these assumptions of restraining adversarial capabilities may drift from appropriate physical model of device-independence, they may provide upper bounds on more capable adversarial models. For the purpose of deriving fundamental limitations, we can accept the trade-off that comes with simplistic assumptions on device-independence protocols. In particular, we can further restrict the adversary such that the device itself is assumed to be \textit{iid}. We define the \textit{iid}-device independent variants $IDI_j$ for $j\in\{0,1,2\}$, where the devices are \textit{iid} and are not allowed memory or communication from one round to the next (e.g., see~\cite{CFH20}).

\begin{definition}
The device-independent secret key agreement (or private) capacity of the device $(\widetilde{\Omega},\rho,\mc{M})$ under the assumption of its \textit{iid} uses assisted with $i$-way communication between allies outside the device and $j$-way communication between the input-output rounds within the device, is given by
\begin{equation}
    \mc{P}^{IDI_j}_{i}(\widetilde{\Omega},\rho,\mc{M})\coloneqq \inf_{\varepsilon>0}\limsup_{n\to \infty}\mu^{IDI_j,\varepsilon}_{i,n}(\widetilde{\Omega},\rho,\mc{M}),
\end{equation}
where $\mu^{IDI_j,\varepsilon}_{i,n}(\rho,\widetilde{\Omega},\mc{M})$ is the maximum key rate optimized over all viable privacy protocols over the \textit{iid} uses of device, while also including a minimization over the possible \textit{iid} devices $IDI_j$ that are compatible with the honest device. We have
\begin{align}
   & \mu^{IDI_j,\varepsilon}_{i,n}(\widetilde{\Omega},\rho,\mc{M}) \nonumber \\
   \coloneqq & \sup_{\hat{P}\in cLOPC_{i}}\inf_{\substack{(\mc{N},\widetilde{\Omega}',\sigma)\in IDI_j\\ (\mc{N},\widetilde{\Omega}',\sigma)\approx_{\varepsilon} (\mc{M},\widetilde{\Omega},\rho)}}\kappa^{\varepsilon}_n(\hat{P},(\mc{N},\widetilde{\Omega}',\sigma)),
\end{align}
where $\kappa^\varepsilon_n$ is the rate of achieved $\varepsilon$-perfect key and classical labels from local classical operations in $\hat{P}\in cLOPC_{i}$ are possessed by the allies (Alice and Bob).
\end{definition}

\begin{definition}
The device-independent capacities $\mc{P}^{DI_j}_{i}$ and $\mc{P}^{IDI_j}_i$ of a bipartite distribution channel $\widetilde{\Omega}$ for the device $DI_{i}$ and $IDI_j$, respectively, are defined as
\begin{align}
\mc{P}^{DI_j}_{i}(\widetilde{\Omega})& \coloneqq \sup_{\rho,\mc{M}}    \mc{P}^{DI_j}_{i}(\widetilde{\Omega},\rho,\mc{M}),\\
\mc{P}^{IDI_j}_{i}(\widetilde{\Omega})& \coloneqq \sup_{\rho,\mc{M}}     \mc{P}^{IDI_j}_{i}(\widetilde{\Omega},\rho,\mc{M}).
\end{align}
\end{definition}
\begin{remark}
We note that 
\begin{equation}
     \mc{P}^{DI_j}_{i}(\widetilde{\Omega},\rho,\mc{M}) \leq  \mc{P}^{IDI_j}_{i}(\widetilde{\Omega},\rho,\mc{M})
\end{equation}
as 
\begin{equation}
    \mu^{DI_j,\varepsilon}_{i,n}(\widetilde{\Omega},\rho,\mc{M})\leq \mu^{IDI_j,\varepsilon}_{i,n}(\widetilde{\Omega},\rho,\mc{M}).
\end{equation}
\end{remark}

Another direct consequence of the \textit{iid} device-independence assumptions is the following lemma.
\begin{lemma}
For any two \textit{iid} devices, $(\widetilde{\Omega},\rho,\mc{M})$ and $(\widetilde{\Omega}',\sigma,\mc{N})$, that are $IDI_j$ and compatible to each other, we have
\begin{align}
    & (\widetilde{\Omega}',\sigma,\mc{N})= (\widetilde{\Omega},\rho,\mc{M}) \nonumber\\ & \quad \implies  \mc{P}^{IDI_j}_{i}(\widetilde{\Omega}',\sigma,\mc{N})=\mc{P}^{IDI_j}_{i}(\widetilde{\Omega},\rho,\mc{M}).
\end{align}
\end{lemma}

The aforementioned definitions are more realistic variant and generalization of the definitions presented in Ref.~\cite{CFH20}. We obtain the definitions in Ref.~\cite{CFH20} if we assume bipartite distribution channel $\widetilde{\Omega}=\id_{A'\to A}\otimes\Lambda_{B'\to B}$ and restrict minimization over compatible devices consisting channel of the form $\widetilde{\Omega}'=\id \otimes\Lambda'$, where $\id$ denotes the identity channel, in the definitions aforementioned.

The main objective of a device-dependent private protocol is to distribute secret keys between two or more trusted allies over quantum channels in the presence of a quantum eavesdropper (e.g., see~\cite{DBWH19}). Traditionally, the secret key agreement between Alice and Bob is over $\id_{A'\to A'}\otimes\Lambda_{A\to B}$, where notion is that Alice transmits a part of composite system in joint state over channel $\Lambda_{A\to B}$ to Bob. It is assumed that the system $A'$ doesn't undergo noisy evolution. Alice and Bob are allowed to use channels $n$ times and make use of adaptive strategy by interleaving each call of channel with $LOPC_i$. In the end of the protocol, Alice and Bob perform $LOPC_i$ to distill secret key between them. However, in practice, even local systems with Alice could undergo noisy quantum evolution. Therefore, we consider quantum/private communication over bipartite quantum distribution channels of the form $\Lambda^1\otimes\Lambda^2$ rather than bipartite quantum distribution channels of the form $\id\otimes\Lambda^2$.

Let us consider a device dependent quantum communication protocol where the goal is for a relay station to transfer prepared entangled state to two allies, Alice and Bob, such that Alice and Bob can distill secret keys between themselves, which is secure from a quantum eavesdropper and the relay station. We can assume that an arbitrary bipartite state $\rho_{A'B'}$ is available at relay station to Charlie. Charlie may operate bipartite quantum distribution channel $\Omega_{A'B'\to A''B''}$ on the state $\rho_{A'B'}$, and is an untrusted party. Charlie then transmits quantum systems $A'', B''$ in the joint state $\Omega(\rho)$ to trusted allies Alice and Bob, respectively, over quantum channels $\Lambda^1_{A"\to A},\Lambda^2_{B''\to B}$. In general, all three parties can make perform $i'$-way LOPC ($LOPC_{i'}$) among themselves in an adaptive strategy, where $i'\in\{0,1,2\}$. Charlie can make $n$ uses of channel $\widetilde{\Omega}=\Lambda^1\otimes\Lambda^2\circ\Omega$ interleaved with $LOPC_{i'}$ between each round (i.e., each call of the channel). At the end of the protocol, the goal is for Alice and Bob to get the state from which secret key is readily accessible upon local measurements. The secret key distillable at the end of Alice and Bob can be $\varepsilon$-close to the ideal secret key.
\begin{definition}
The device-dependent privacy distribution capacity $\mc{R}_{i'}$ over a bipartite quantum distribution channel $\widetilde{\Omega}$ assisted with $i'$-way communication ($LOPC_{i'}$) among Charlie, Alice, and Bob for $i'\in\{0,1,2\}$ is defined as
\begin{equation}
    \mc{R}_{i'}(\widetilde{\Omega})\coloneqq \inf_{\varepsilon>0} \limsup_{n\to\infty}\nu^{\varepsilon}_{i',n}(\widetilde{\Omega}), 
\end{equation}
where $\nu^{\varepsilon}_{i',n}(\widetilde{\Omega})$ is the maximum $\varepsilon$-perfect key rate obtained among all possible repeatable privacy protocols (assisted with $i'$-way communication $LOPC_{i'}$ among the relay station and the trusted allies) that uses channel $\widetilde{\Omega}$ $n$ times.

Device-dependent privacy distribution capacity over $\id_{A'\to A'}\otimes\Lambda_{A\to B}$, where Alice herself is at the relay station and sender to Bob, with $LOPC_i$ assistance reduces to device-dependent $LOPC_i$-assisted private capacity $\mc{P}_i^{DD}(\Lambda)$ over point-to-point quantum channel $\Lambda$.
\end{definition}

\begin{observation}
The device-dependent privacy distribution capacity $R_i$ over a bipartite quantum distribution channel $\widetilde{\Omega}=\Lambda^1\otimes\Lambda^2$ is upper bounded by the device-dependent private capacity over point-to-point channels $\Lambda^1$ and $\Lambda^2$, i.e.,
\begin{equation}
    \mc{R}_{i}(\widetilde{\Omega})\leq \min\{\mc{P}_i^{DD}(\Lambda^1), \mc{P}_i^{DD}(\Lambda^2)\}.
\end{equation}
The protocol for the privacy distribution over $\Lambda^1\otimes\Lambda^2$ reduces to secret key agreement protocols over $\Lambda^1$ or $\Lambda^2$ if we assume $\Lambda^2=\id$ with Bob at relay station as sender or $\Lambda^1=\id$ with Alice at the relay station as the sender, respectively. Under such reduction of the protocol, lesser amount of information is leaked to a quantum eavesdropper as part of a noisy evolution becomes noiseless.
\end{observation}

\begin{lemma}
The device-independent secret key capacities $\mc{P}_{i}^{DI_j}$ and $\mc{P}_{i}^{IDI_j}$ of a device $(\widetilde{\Omega},\rho,\mc{M})$ in terms of optimized privacy distribution capacity $\mc{R}_{\max\{i,j\}}$ are
\begin{align}
    \mc{P}_{i}^{DI_j}(\widetilde{\Omega},\rho,\mc{M})& \leq \inf_{\substack{(\widetilde{\Omega}',\sigma,\mc{N})\in IDI_j\\ (\widetilde{\Omega}',\sigma,\mc{N})= (\widetilde{\Omega},\rho,\mc{M})}}\mc{R}_{\max\{i,j\}}(\widetilde{\Omega}'),\\
     \mc{P}_{i}^{IDI_j}(\widetilde{\Omega},\rho,\mc{M}) & \leq \inf_{\substack{(\widetilde{\Omega}',\sigma,\mc{N})\in IDI_j\\ (\widetilde{\Omega}',\sigma,\mc{N})= (\widetilde{\Omega},\rho,\mc{M})}}\mc{R}_{\max\{i,j\}}(\widetilde{\Omega}').
\end{align}
\end{lemma}
We omit the proof of the above lemma as proof arguments are similar to the proof of \cite[Eqs.~(47) and (49)]{CFH20}.

\begin{lemma}
 The device-independent secret key capacity $\mc{P}^{IDI_j}_{i}$ of an honest device $(\widetilde{\Omega},\rho,\mc{M})$ when $\widetilde{\Omega}_{A'B'\to AB}=\Lambda^1_{A'\to A}\otimes\Lambda^2_{B'\to B}$ is upper bounded by
 \begin{align}
  & \mc{P}_{i}^{IDI_j}(\widetilde{\Omega},\rho,\mc{M}) \nonumber\\ 
  \leq & \inf_{\substack{(\widetilde{\Omega}',\sigma,\mc{N})\in IDI_j\\ (\widetilde{\Omega}',\sigma,\mc{N})= (\widetilde{\Omega},\rho,\mc{M})}} \min\{\substack{\mc{P}^{DD}_{\max\{i,j\}}({\Lambda'}^1), \mc{P}^{DD}_{\max\{i,j\}}({\Lambda'}^2)}\},
 \end{align} 
 where $\widetilde{\Omega}'={\Lambda'}^1\otimes{\Lambda'}^2$. It follows that for all $i,j\in\{0,1,2\}$, we have
 \begin{equation}
      \mc{P}_{i}^{IDI_j}(\widetilde{\Omega},\rho,\mc{M})\leq \min\{\mc{P}^{DD}_2(\Lambda^1),\mc{P}^{DD}_2(\Lambda^2)\}.
 \end{equation}
 Furthermore, if the point-to-point channels $\Lambda^1, \Lambda^2$ are PPT channels, we have
 \begin{equation}
   \mc{P}_{2}^{IDI_j}(\widetilde{\Omega},\rho,\mc{M}) \leq \min\{\substack{\mc{P}^{DD}_2(\Lambda^1), \mc{P}^{DD}_2(\Lambda^2),\\ \mc{P}^{DD}_2(\vartheta\circ\Lambda^1),\mc{P}^{DD}_2(\vartheta\circ\Lambda^2)}\},
 \end{equation}
 where $\vartheta$ is a partial transposition map, i.e., $\vartheta(\rho)=\rho^{\Gamma}$.
\end{lemma}
We now restrict to the case of tele-covariant channels which can be defined as in Refs.~\cite{Bennett1996, Horodecki1991} and first employed in Ref.~ \cite{Pirandola2017} in the context of secret key agreement protocols over point-to-point quantum channels.
\begin{definition}
Let G be a group with unitary representation $g\rightarrow U^g_A$ on $\mathcal{H}_A$ and $g\rightarrow V^g_B$ on $\mathcal{H}_B$. A quantum channel $\Lambda_{A\to B}$ is covariant with respect to the unitary group $\left\{\mathcal{U}_g\right\}_{g\in\mathcal{G}}$ if 
\begin{equation}\label{eq:cov}
 V^g_B\Lambda(\cdot)V^{g\dagger}_B=\Lambda \left(U^g_A(\cdot)U^{g\dagger}_A\right) 
\end{equation}
If a quantum channel $\Lambda_{A\to B}$ satisfies~\eqref{eq:cov} such that the unitary group $G$ is a one-design, i.e.,
\begin{equation}
    \frac{1}{|G|}\sum_{g\in G}U^g_A\rho_AU^{g\dagger}_A=\frac{1}{|A|}\mathbbm{1}_A\qquad \forall \rho_A,
\end{equation}
then $\Lambda_{A\to B}$ is said to be tele-covariant.
\end{definition}

A channel that is covariant with respect to one-design unitaries can be simulated via LOCC and the Choi state of the channel as a shared resource state \cite{Chiribella2009}. That is, 
\begin{equation}
    \Lambda_{A\rightarrow B}(\rho_{A}) = \mc{L}_{AA'\rightarrow B}\left(\Phi^{\Lambda}_{A'B}\otimes \rho_A\right),
\end{equation}
where $\mc{L}_{AA'\to B}$ is an LOCC channel, with the classical communication being from $A$ to $B$ and $\Phi^{\Lambda}_{A'B}\coloneqq \Lambda(\Phi^+_{A'A})$ is the Choi state of the channel. The above equation informally implies that any quantum communication via the channel $\Lambda$ is equivalent to sharing the Choi state $\Phi^{\Lambda}$ followed by local operations and classical communication. 
The Lemma below follows from the observation that the private capacity of tele-covariant channels are upper bounded by the relative entropy of entanglement of the Choi state of the channel \cite{Pirandola2017}.
\begin{lemma}
Consider a tele-covariant distribution channel $\widetilde{\Omega}_{A'B'\to AB}=\Lambda^1_{A'\to A}\otimes\Lambda^2_{B'\to B}$, where both point-to-point channels $\Lambda^1$ and $\Lambda^2$ are tele-covariant. The device-independent secret key capacity $\mc{P}^{IDI_j}_{i}$ of an honest device $(\widetilde{\Omega},\rho,\mc{M})$ with such a tele-covariant distribution channel $\widetilde{\Omega}_{A'B'\to AB}=\Lambda^1_{A'\to A}\otimes\Lambda^2_{B'\to B}$ is upper bounded by
 \begin{align} 
  & \mc{P}_{i}^{IDI_j}(\widetilde{\Omega},\rho,\mc{M}) \nonumber\\ 
  \leq & \inf_{\substack{(\widetilde{\Omega}',\sigma,\mc{N})\in IDI_j\\ (\widetilde{\Omega}',\sigma,\mc{N})= (\widetilde{\Omega},\rho,\mc{M})}} \min\{{E_R(\Phi^{\Lambda'_1}),E_R(\Phi^{\Lambda'_2})}\},\label{lem:choi-states-channels}
 \end{align} 
 where $\widetilde{\Omega}'={\Lambda'}^1\otimes{\Lambda'}^2$ and $\Phi^{\Lambda'_1}$ and $\Phi^{\Lambda'_2}$ are the Choi states of the channels $\Lambda'^1$ and $\Lambda'^2$, respectively.
\end{lemma}

\subsubsection{Some practical prototypes}\label{sup:prototypes}
Let us now focus on three widely considered noise models for the qubit systems: dephasing channel $\mc{P}^p$, depolarizing channel $\mc{D}^p$, and erasure channel $\mc{E}^p$. The actions of these tele-covariant channels on the density operators $\rho$ of a qubit system are given as
\begin{enumerate}
\item Dephasing channel:
    $\mathcal{P}^p(\rho)\equiv (1-p)\rho+p \sigma_Z\rho \sigma_Z$, where $\sigma_Z$ is the Pauli-Z operator.
    \item Depolarizing channel: $\mathcal{D}^{p} (\rho)\equiv (1-p)\rho+p\frac{1}{2}\mathbbm{1}$.
    \item Erasure channel: $\mathcal{E}^{p} (\rho)\equiv p\rho+(1-p)\op{e}$, where $\op{e}$ is the erasure symbol, orthonormal to the support of the input state. 
\end{enumerate}

Let us now consider that Alice and Bob carry out the CHSH protocol over the channel $\mathrm{id}_{A'\rightarrow A}\otimes \Lambda_{B'\rightarrow B}$. As discussed above, for CHSH protocols, the relevant statistics are the CHSH violation $\omega(\rho,\mathcal{M})$ and QBER $P_{\textrm{err}}(\rho,\mathcal{M})$. Thus, for CHSH protocols, the infimum in \eqref{lem:choi-states-channels} is reduced to the tuples $\left(\tilde{\Omega'},\sigma,\mathcal{N}\right)$ that satisfy the CHSH statistics. 

We first consider the honest device $ \left(\rho,\textrm{id}_{A'\rightarrow A}\otimes\mathcal{D}^p_{B'\rightarrow B},\mathcal{M}\right)$. Let $\mathcal{M}$ be an arbitrary but honest measurements considered in the CHSH protocol. Then, the CHSH violation observed by Alice and Bob is given as

\begin{align}
    \omega(\mathcal{D}^p(\rho),\mathcal{M})&=  (1-p)\omega(\rho,\mathcal{M})+p\omega(\frac{1}{4}\mathbbm{1},\mathcal{M})\\
    &= (1-p)\omega(\rho,\mathcal{M}).\\
    &\leq (1-p) \max_{\mathcal{M}}\omega(\rho,\mathcal{M})\equiv \omega^{\star}
\end{align}
Here, $\omega(\mathcal{D}^p(\rho),\mathcal{M})$ corresponds to the CHSH violation observed from the statistics obtained if the state $\mathcal{D}^{p}(\rho)$ is measured by $\mathcal{M}$. The second equality follows from $\omega(\mathcal{M},\frac{1}{4}\mathbbm{1})=0$.  
 The QBER associated with the state is $(1-p)P_{\textrm{err}}(\rho,\mathcal{M})+\frac{1}{2} p$. We thus obtain the limits on the statistics that Alice and Bob would obtain on carrying out a CHSH protocol with the device $\left(\textrm{id}_{A'\rightarrow A}\otimes\mathcal{D}^p_{B'\rightarrow B}, \rho,\mathcal{M}\right)$. 
 
 We now construct a strategy $\left(\textrm{id}_{A'\rightarrow A}\otimes\mathcal{P}^q_{B'\rightarrow B}, \Phi,\mathcal{M}\right)$, where $\Phi$ is a maximally entangled state. By appropriate choice of parameter $q$ and the measurements $M^0_a,M^1_a,M^1_b,M^2_b$, we can replicate the CHSH violation and $P_{\textrm{err}}$ obtained from carrying out a CHSH protocol with the device $\left(\textrm{id}_{A'\rightarrow A}\otimes\mathcal{D}^p_{B'\rightarrow B}, \rho,\mathcal{M}\right)$. The noise of the dephasing channel $q$ is chosen as $\frac{1-C}{2}$, where $C=\sqrt{(\frac{\omega^*}{2})^2-1}$. The measurements are given as $M^{1,2}_a=\frac{\sigma_z}{\sqrt{1+C^2}}\pm \frac{C}{\sqrt{1+C^2}}\sigma_x$, $M^1_b=\sigma_z$, $M^2_b=\sigma_x$. With this strategy, we obtain the Bell violation $\omega^{\star}$. For replicating the statistics of QBER, with prob $P_{\textrm{err}}$, Alice chooses $A_0=\sigma_z$, else she randomly chooses a bit. This strategy has been previously used in Ref.~\cite{Pironio2009} to show tightness of the obtained lower bounds for one-way CHSH protocols. With this strategy, we can replicate the statistics obtained from CHSH protocols performed over depolarizing channels. 

Combining the above observations with \eqref{lem:choi-states-channels}, we then obtain the following: 
\begin{align}
     \mc{P}_{i}^{IDI_j}(\textrm{id}_A\otimes \mathcal{D}^p)
  &\leq \min\left\{E_R(\mathcal{P}^{\frac{1-C}{2}}\left(\Phi\right)),E_R(\mathcal{D}^p(\Phi))\right\}\\
  & \leq \min\left\{1-H\left(\frac{1-C}{2}\right),1-H\left(\frac{3p}{4}\right)\right\}\label{eq:upper-bound-dep1}
\end{align}
We also see that the maximum CHSH violation $\omega^{\star}$ obtained from the depolarizing channel with noise $p$ is $(1-p)2\sqrt{2}$. Substituting $C=\sqrt{\left(\frac{(1-p)2\sqrt{2}}{2}\right)^2-1}$ in \eqref{eq:upper-bound-dep1}, we obtain
\begin{multline}
     \mc{P}_{i}^{IDI_j}(\textrm{id}_A\otimes \mathcal{D}^p)\\
  \leq \min\left\{1-H\left(\frac{1}{2} (1 - \sqrt{1 - 4 p + 2 p^2})\right),1-H(3p/4)\right\}.
\end{multline}
We thus obtain that for CHSH protocols, the DI secret key capacity of depolarizing channels is strictly less than the 
private capacity of the depolarizing channel as can be seen in Figure~\ref{fig:channel_capacities}. Here, we have used the upper bounds on the two-way (LOCC-assisted) device-dependent secret-key-agreement capacities for depolarizing and dephasing qubit channels obtained in Ref.~\cite{Pirandola2017}.

Next, let us consider the erasure channel. Let $\mathcal{M}$ be a set of four measurements. The CHSH violation observed is 
\begin{align}
    \omega(\mathcal{E}^p(\rho),\mathcal{M})&= (1-p)\omega(\rho,\mathcal{M})+p\omega(\op{e},\mathcal{M})\\
    &= (1-p)\omega(\rho,\mathcal{M}).\\
    &\leq (1-p) \max_{\mathcal{M}}\omega(\rho,\mathcal{M})\equiv \omega^{\star}
\end{align}
We thus see that the upper bound on CHSH violation by erasure channel of noise $p$ and the depolarizing channel $p$ are exactly the same. However, the QBER obtained by carrying out the CHSH protocol across the erasure channel and depolarizing channel can be different. We then observe that any value of QBER can be observed by changing the measurement setting with the dephasing channel. For erasure channels the QBER is one, so we choose the measurements settings such that the QBER obtained from dephasing channel is also one. We thus obtain 
\begin{align}
     &\mc{P}_{i}^{IDI_j}(\textrm{id}_A\otimes \mathcal{E}^p)\\
  &\leq \min\left\{E_R(\mathcal{P}^{\frac{1-C}{2}}\left(\Phi\right)),E_R(\mathcal{E}^p(\Phi))\right\}\\
  & \leq \min\left\{1-H\left(\frac{1-C}{2}\right),1-p\right\}\label{eq:upper-bound-dep}\\
   &\leq \min\left\{1-H\left(\frac{1}{2} (1 - \sqrt{1 - 4 p + 2 p^2})\right),1-p\right\}.
\end{align}

That is, for CHSH protocols, the DI secret key capacity of erasure channels is strictly less than the private capacity of the erasure channels as can be seen in Figure~\ref{fig:channel_capacities}. 

It is interesting to observe that in the above analyses, the violation of the CHSH inequality had a vital role in limiting the CHSH DI capacity across various example channels. Due to the structure of the dephasing channel, which was the attacking channel, the QBER did not end up influencing the upper bounds.  

\begin{figure}
    \centering
    \includegraphics[width=\linewidth]{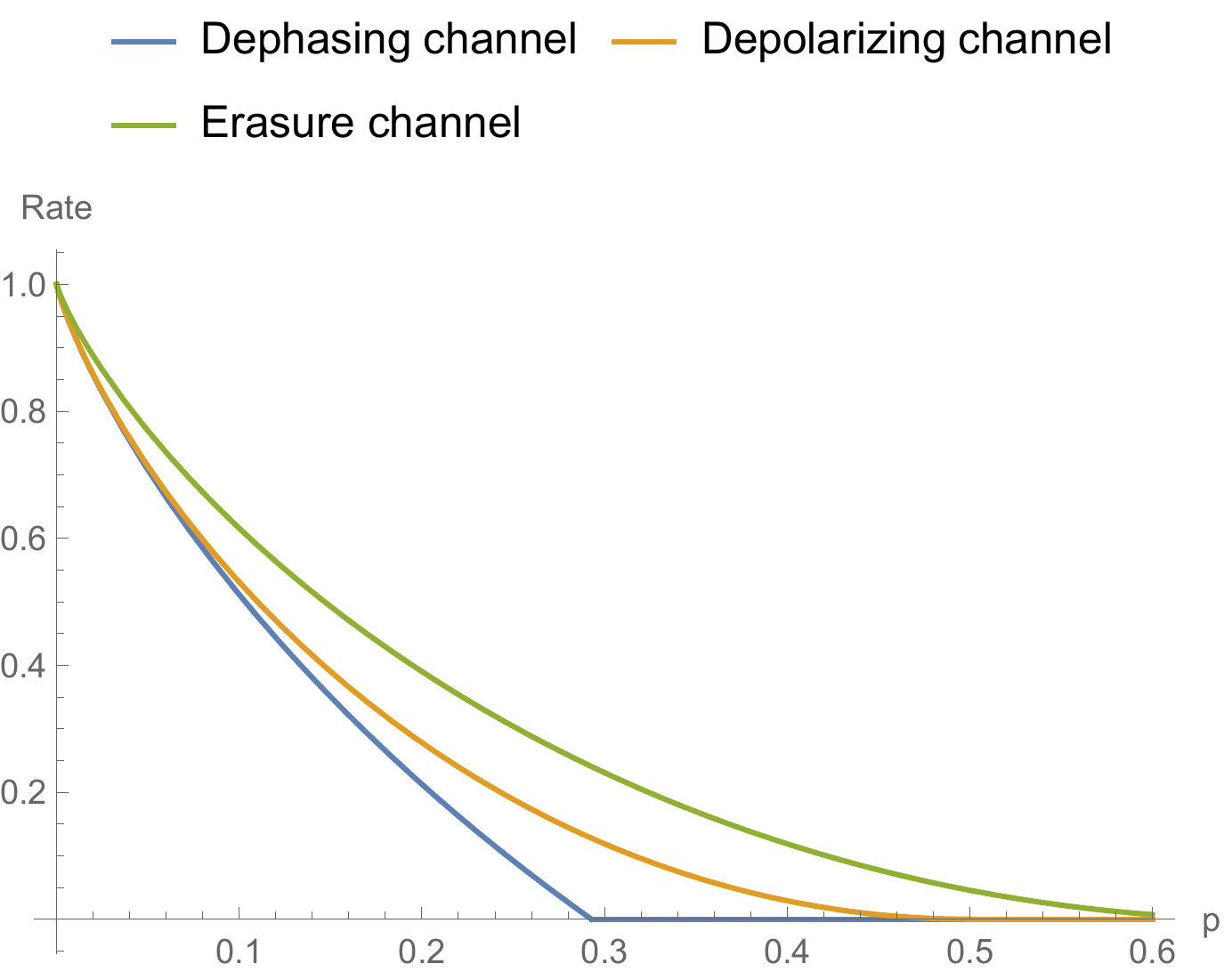}
    \caption{In the above figure, we plot upper bounds on the device-dependent QKD capacities of depolarizing channel (yellow line), dephasing channel (blue line) and erasure channel (green line). We notice that the upper bounds for erasure and dephasing channels are achievable device-dependent QKD rates (capacities). We then notice that for the CHSH protocols, the upper bounds on the DI-QKD capacities of channels is limited by the device-dependent QKD capacity of dephasing channels.}
    \label{fig:channel_capacities}
\end{figure}

\subsection{Upper bound via cc-squashed entanglement}
In what follows, we will define an entanglement measure that takes as an input a bipartite state and a pair $M_{AB}$ of POVMs $\{M_a\}_a$ and $\{M_b\}_b$ which act on systems locally, $M_{AB}\coloneqq M_a\otimes M_b$.
\begin{definition}
A cc-squashed entanglement of a bipartite state $\rho_{AB}$  reads
$E_{sq}(\rho_{AB}, {\mathrm M})$ is defined as follows:
\begin{equation}
    E_{sq}^{cc}(\rho,{\mathrm M}) := \inf_{\Lambda:E\rightarrow E'} I(A:B|E')_{{\mathrm M}_{AB}\otimes \Lambda_E \psi^{\rho}_{ABE}}
    \label{eq:ccsq}
\end{equation}
where ${\mathrm M}_{AB}$ is a pair of POVMs ${\mathrm M}= M_a\otimes M_b$, and $\psi^\rho_{ABE}\coloneqq \op{\psi^\rho}_{ABE}$ is a state purification of $\rho_{AB}$.
\end{definition}
Following Ref.~\cite{CW04}, we observe first
that 
\begin{observation}
\label{obs:equiv}
For a bipartite state $\rho_{AB}$ and
a pair of POVMs ${\mathrm M}= M_a\otimes M_b$, there is 
\begin{equation}
    E_{sq}^{cc}(\rho_{AB},{\mathrm M})=\inf_{\rho_{ABE}={\mbox Ext}(\rho_{AB})} I(A:B|E)_{{\mathrm M}\otimes \id_E \rho_{ABE}}
    \label{eq:ccext}
\end{equation}
where $\mbox{Ext}(\rho_{AB})$ is an arbitrary state extension of $\rho_{AB}$ to system $E$, i.e., $\rho_{ABE}$ is a state such that $\Tr_E[\rho_{ABE}]=\rho_{AB}$.
\end{observation}
\begin{proof}
Following Ref.~\cite{CW04}:
To see that $\eqref{eq:ccsq} \leq \eqref{eq:ccext}$ we note that every
extension can be obtained from the purifying system by an appropriate channel. Indeed, we first note that
$|\psi\>_{ABEE''}$ which purifies $\rho_{ABE}$ is related by an isometry
to any state purification $|\psi^{\rho}\>$ of $\rho_{AB}$. Hence, a channel performing this isometry and tracing out $E''$ generates an  extension $\rho_{ABE}$. Thus, the infimum in \eqref{eq:ccsq} which varies over $\Lambda_E$ acting on $\psi^{\rho}$
can be seen as optimisation over the set of arbitrary extensions measured by ${\mathrm M}$ on $AB$, as it is the case in \eqref{eq:ccext}. Note that we have used the fact that measurements ${\mathrm M}$ are the same in both formulas, and
the extension in \eqref{eq:ccext} is taken before the measurement.

Conversely,  we have also $\eqref{eq:ccsq} \geq \eqref{eq:ccext}$, because application of a channel on system $E$ of a purified state $\psi^{\rho}$, results in an extension $\rho_{ABE'}$.
\end{proof}
Owing to the above observation, we 
can see that the cc-squashed entanglement is convex, as stated in the lemma below:
\begin{lemma}
For a pair of measurements ${\mathrm M}$, two states $\rho_{AB}$ and $\rho_{AB}'$, $0<p<1$, there is 
\begin{equation}
    E_{sq}^{cc}(\bar{\rho}_{AB},{\mathrm M})\leq pE_{sq}^{cc}(\rho_{AB},{\mathrm M})+(1-p)E_{sq}^{cc}(\rho_{AB}',{\mathrm M}),
\end{equation}
where  $\bar{\rho}_{AB}= p\rho_{AB}+(1-p)\rho_{AB}'$.
\label{lem:conv}
\end{lemma}
\begin{proof}
Consider first two tripartite extensions of the form:
\begin{align}
    &\rho_1:= \rho_{ABE}, \\
    &\rho_2:= \rho_{ABE}'
\end{align}
of $\rho_{AB}$ and $\rho_{AB}'$ respectively.
Consider then the state of the following form:
\begin{equation}
\rho_{ABEF}={\mathrm M}\otimes \id_{EF}(p\rho_1\otimes \op{0}_F + (1-p)\rho_2\otimes \op{1}_F).
\end{equation}
Note that it is measured 
 extension of the state $\bar{\rho}_{AB}$. Indeed, by linearity of the partial trace, tracing out  over systems $F$ and $E$ we obtain the $p$-weighted mixture of states $\rho_{AB}$ and $\rho_{AB}'$, measured by ${\mathrm M}$, which is the measured state $\bar{\rho}_{AB}$. 
 
 By Observation \ref{obs:equiv} we can use the definition of $E_{sq}^{cc}$ based on extensions rather than channels. In what follows we go along similar lines to Refs.~\cite{WDH19,CW04}.
 \begin{align}
 &E_{sq}^{cc}(\bar{\rho}_{AB}, {\mathrm M})\nonumber\\&=\inf_{\rho_{ABE}={\mbox Ext}(\bar{\rho}_{AB})} I(A:B|E)_{{\mathrm M}\otimes \id_E \rho_{ABE}}\nonumber\\ &\leq 
 I(A:B|EF)_{\rho_{ABEF}}\\&= p  I(A:B|E)_{M\otimes \id_E \rho_1}+(1-p)
 I(A:B|E)_{M\otimes \id_E \rho_2}
 \label{eq:separate}
 \end{align}
 In the above we first narrow the infimum to a particular extension $\rho_{ABEF}$.
 The equality follows from the fact that system $F$ is classical, and conditioning over such a system yields average value of the conditional mutual information. We also have used linearity of measurement ${\mathrm M}$:
 \begin{align}
     \rho_{ABEL}=& p M\otimes \id_{EF}\rho_1\otimes \op{0} +\nonumber\\
     &\quad (1-p) M\otimes \id_{EF}\rho_2\otimes \op{1},
 \end{align}
 to separate terms in (\ref{eq:separate}).
 Since the extensions $\rho_1$ and $\rho_2$ were arbitrary, we can
 also take infimum over them, obtaining:
 \begin{align}
 &E_{sq}^{cc}(\bar{\rho}_{AB}, {\mathrm M})\leq 
  p \inf_{\rho_{ABE}=\mbox{Ext}(\rho_{AB})} I(A:B|E)_{M\otimes \id_E \rho_{ABE}}+\nonumber\\&(1-p)
 \inf_{\rho'_{ABE}=\mbox{Ext}(\rho_{AB}')}I(A:B|E)_{M\otimes \id_E \rho'_{ABE}}
 \label{eq:separate1}
 \end{align}
Again by Observation \ref{obs:equiv} on the RHS we have
$pE_{sq}^{cc}(\rho_{AB}, {\mathrm M})+(1-p)E_{sq}^{cc}(\rho_{AB}', {\mathrm M})$, hence the assertion follows.
\end{proof}
Following \cite[Theorem~3.5]{CEHHOR} we obtain
\begin{theorem}[\cite{CEHHOR}]
For a bipartite state $\rho$, it's purified state $\psi^{\rho}$, and a pair of measurements ${\mathrm M}$, there is
\begin{equation}
K_{DD}({\mathrm M}\otimes \id \psi^{\rho}) \leq E_{sq}^{cc}(\rho,{\mathrm M})
\end{equation}
\label{thm:Matthias}
\end{theorem}
\begin{proof}
The proof boils down to invoking \cite[Theorem 3.5]{CEHHOR} for 
a tripartite ccq state $\rho_{ccq}:={\mathrm M}\otimes \id \psi^{\rho}$, and
noticing that $I(A:B\downarrow E)_{\rho_{ccq}} = E_{sq}^{cc}(\rho,{\mathrm M})$.
\end{proof}
To see the application of the 
cc squashed entanglement we need the
following fact.
\begin{lemma}
The \textit{iid} quantum device independent key achieved by protocols using single pair of measurements $(\hat{x},\hat{y})$ applied to ${\cal M}$ of a device $(\rho,{\cal M})$, is upper bounded as follows:
\begin{align}
    &K_{DI,dev}^{iid,(\hat{x},\hat{y})}(\rho,{\cal M}):=\nonumber\\
    &\inf_{\epsilon>0}\limsup_n\sup_{{\cal P}\in LOPC}\inf_{{(\sigma,{\cal N})\approx_{\epsilon}(\rho,{\cal M})}} \kappa^{\epsilon,(\hat{x},\hat{y})}_n({\cal P}({\mathrm N}(\sigma)^{\otimes n}) \nonumber\\\leq 
    & \inf_{{(\sigma,{\cal N})=(\rho,{\cal M})}} K_{DD}({\cal N}(\hat{x},\hat{y})\otimes \id \psi^\sigma),
\end{align}
where ${\mathrm N}\equiv{\cal N}(\hat{x},\hat{y})$ is a single pair of measurements induced by inputs $(\hat{x},\hat{y})$ on $\cal N$ and  where $\kappa_n^{\varepsilon(\hat{x},\hat{y})}$ is the rate of achieved $\varepsilon$-perfect key and classical labels from local classical operations in $\hat{P}\in cLOPC$ are possessed by the allies (Alice and Bob).
\label{lem:ubound_by_key}
\end{lemma}
\begin{proof}
\begin{align}
    &K_{DI,dev}^{iid,(\hat{x},\hat{y})}(\rho,{\cal M}) =\nonumber\\&\inf_{\epsilon>0}\limsup_n\sup_{{\cal P}\in LOPC}\inf_{{(\sigma,{\cal N})\approx_{\epsilon}(\rho,{\cal M})}} \kappa^{\epsilon,(\hat{x},\hat{y})}_n({\cal P}({\mathrm N}(\sigma)^{\otimes n})) \\ &\leq
    \inf_{\epsilon>0}\inf_{{(\sigma,{\cal N})\approx_{\epsilon}(\rho,{\cal M})}} \limsup_n\sup_{{\cal P}\in LOPC}\kappa^{\epsilon,(\hat{x},\hat{y})}_n({\cal P}({\mathrm N}(\sigma)^{\otimes n}))\\ &\leq 
    \inf_{{(\sigma,{\cal N})=(\rho,{\cal M})}} \inf_{\epsilon>0}\limsup_n\sup_{{\cal P}\in LOPC}\kappa^{\epsilon,(\hat{x},\hat{y})}_n({\cal P}({\mathrm N}(\sigma)^{\otimes n})) \nonumber\\
    &=\inf_{{(\sigma,{\cal N})=(\rho,{\cal M})}} K_{DD}({\cal N}(\hat{x},\hat{y})\otimes \id \psi^\sigma).
\end{align}
In the above we first use the max-min inequality for sup and limsup \cite{CFH20}. We then narrow infimum to devices that ideally mimic the device $(\rho,{\cal M})$. We further notice that the key $\kappa^{\epsilon,(\hat{x},\hat{y})}_n({\cal P}({\mathrm N}(\sigma)^{\otimes n})$ where the supremum is taken over LOPC protocols
equals the device dependent key of a tripartite ccq state ${\mathrm N}\otimes \id \psi_\sigma$. 
\end{proof}
The LOPC protocols considered in the above definition consists of the error correction and parameter amplification in the DI-QKD protocols over the ccq state $\mathcal{N}(\hat{x},\hat{y})\otimes \textrm{id}\psi^{\sigma}$. We have assumed that the test rounds and the key generation rounds are known to Eve due to classical communication carried out by Alice and Bob. We should specify that the distinction in the rounds was not known prior to the preparation of the device. This knowledge becomes available to the eavesdropper after Alice and Bob have performed the measurements and classically communicated with each other. This extra knowledge of distinction between test rounds and key generation rounds is instrumental in obtaining tighter upper bounds for DI-QKD protocols.

Due to Theorem~\ref{thm:Matthias} and the above lemma we have immediate corollary:
\begin{corollary}
The iid quantum device independent key achieved by protocols using single pair of measurements $(\hat{x},\hat{y})$ applied to $ {\cal M}$ of a device $(\rho,{\cal M})$, is upper bounded as follows:    
\begin{align}
    K_{DI,dev}^{iid,(\hat{x},\hat{y})}(\rho,{\cal M})&\leq \inf_{{(\sigma,{\cal N})\equiv(\rho,{\cal M})}}E_{sq}^{cc}(\sigma,{\cal N}(\hat{x},\hat{y}))\\
    &=:E_{sq,dev}^{cc}(\rho,{\cal M}({\hat{x},\hat{y}})).
\end{align}
\label{cor:reduced_dev}
\end{corollary}

\begin{observation}
If two quantum devices $(\rho,\mc{M})$ and $(\sigma,\mc{N})$ are such that $(\rho,\mc{M})=(\sigma,\mc{N})$ then
\begin{align}
    K_{DI,dev}^{iid,(\hat{x},\hat{y})}(\rho,{\cal M})& =K_{DI,dev}^{iid,(\hat{x},\hat{y})}(\sigma,{\cal N}),\\
    E_{sq,dev}^{cc}(\rho,{\cal M}({\hat{x},\hat{y}}))&=E_{sq,dev}^{cc}(\sigma,{\cal N}({\hat{x},\hat{y}}))
\end{align}
for any valid choice of $(\hat{x},\hat{y})$.
\end{observation}

\color{black}
We now pass to study the upper bounds provided in Refs.~\cite{AL20,FBL+21}. We first note that in Ref.~\cite{AL20} conditions of equal CHSH value $\omega$ and QBER $P_{err}$ are considered instead of equality of attacking and honest device. It is straightforward to adopt the above corollary to this case:
\begin{corollary}
The iid quantum device independent key achieved by protocols using single pair of measurements $(\hat{x},\hat{y})$ applied to ${\cal M}$ of a device $(\rho,{\cal M})$, is upper bounded as follows:    
\begin{align}
    &K_{DI,par}^{iid,({\hat{x},\hat{y}})}(\rho,{\cal M}) 
    \coloneqq \nonumber\\
    &\inf_{\epsilon>0}\limsup_n\sup_{{\cal P}\in LOPC}\inf_{\underset{P_{err}(\sigma,{\cal N})\approx_\epsilon P_{err}(\rho,{\cal M})}{\omega(\sigma,{\cal N})\approx_\epsilon\omega(\rho,{\cal M})}} \kappa^{\epsilon,(\hat{x},\hat{y})}_n({\cal P}({\mathrm N}(\sigma)^{\otimes n}))\label{eq:di_par1}\\
    &\leq  \inf_{\underset{P_{err}(\sigma,{\cal N})= P_{err}(\rho,{\cal M})}{\omega(\sigma,{\cal N})=\omega(\rho,{\cal M})}}E_{sq}^{cc}(\sigma,{\mathrm N}) =: E_{sq,par}^{cc}(\rho,{\cal M}(\hat{x},\hat{y})),
\end{align}
where $N={\cal N}(\hat{x},\hat{y})$ is a single pair of measurements induced by inputs $(\hat{x},\hat{y})$ on $\cal N$.
\label{col:esqbound}
\end{corollary}

\begin{observation}
If two quantum devices $(\rho,\mc{M})$ and $(\sigma,\mc{N})$ are such that $\omega(\rho,\mc{M})=\omega(\sigma,\mc{N})$ and $P_{err}(\rho,\mc{M})=P_{err}(\sigma,\mc{N})$ then
\begin{align}
   K_{DI,par}^{iid,({\hat{x},\hat{y}})}(\rho,{\cal M})& =K_{DI,par}^{iid,({\hat{x},\hat{y}})}(\sigma,{\cal N}),\\
    E_{sq,par}^{cc}(\rho,{\cal M}(\hat{x},\hat{y}))&=E_{sq,par}^{cc}(\sigma,{\cal N}(\hat{x},\hat{y})),
\end{align}
for any valid choice of $(\hat{x},\hat{y})$. This is to say that $E^{cc}_{sq, par}(\rho,\mc{M}(\hat{x},\hat{y}))$ is a function explicitly depending on only two parameters $\omega$ and $P_{err}$.
\end{observation}

\color{black}
The quantity defined in \eqref{eq:di_par1} $K_{DI,par}^{iid,(\hat{x},\hat{y})}$ depends on the choice of Bell inequality and its violation $\omega$ and quantum bit error rate $P_{err}$.
For further considerations one can assume that $P_{err}$ is computed as $P(a\neq b |\hat{x},\hat{y})$, as the key is generated by $(\hat{x},\hat{y})$, however $\omega$ remains a free parameter.
In any case, not to overload the notation, we refrain from decorating definition of $K_{DI,par}^{iid,(\hat{x},\hat{y})}$ by $\omega$, and make $\omega$ explicitly known from the context if needed (e.g. see Theorem~\ref{thm:convexification}).

We argue now, that the $E^{cc}_{sq,par}(\rho,{\cal M}(\hat{x},\hat{y}))$ is equal to the bound \cite[Eq.~(19)]{AL20}. Before that we invoke the notation of \cite{AL20} where 
$(\sigma_{ABE},{\cal M}(\hat{x},\hat{y}))\in \hat{\Sigma}(\omega^*,Q^*)$ iff
$\tr_E(\sigma_{ABE},{\cal M}(\hat{x},\hat{y})\otimes \id_E)=q(ab|\hat{x}\hat{y})$ where $\omega(q(ab|\hat{x}\hat{y}))=\omega^* $ and $P_{err}(q(ab|\hat{x}\hat{y}))=Q^*$.

The bound of \cite{AL20} reads
\begin{equation}
    {\cal I}_{par}(\omega^*,Q^*,\hat{x},\hat{y}):=\inf_{\sigma \in \hat{{\Sigma}}(\omega^*,Q^*)}I(A;B\downarrow E)_{\sigma(\hat{x},\hat{y})}
\end{equation}
where the quantity $I(A:B\downarrow E):=\inf_{\Lambda:E\rightarrow E'} I(A:B|E')_{\sigma(\hat{x},\hat{y})}$ is 
computed on a state $\sigma_{ABE}$ measured with ${\cal M}(\hat{x},\hat{y})$ on $AB$
 for some measurements ${\cal M}$.
 
The equivalence is encapsulated in the following theorem:

\begin{theorem}
    Let $(\rho, {\cal M})$ be a quantum device with parameters $\omega^*,Q^*$ where $Q^*$ is computed based on the inputs $(\hat{x},\hat{y})$. Then there is
    \begin{equation}
        E_{sq,par}^{cc}(\rho,{\cal M}(\hat{x},\hat{y})) = {\cal I}(\omega^*,Q^*,\hat{x},\hat{y})
    \end{equation}
  for any choice of the inputs $(\hat{x},\hat{y})$.  
\end{theorem}
\begin{proof}
In what follows we can fix $(\hat{x},\hat{y})$ arbitrarily.
We first prove that for any quantum realization $(\rho,{\cal M})$ of a device with parameters $\omega^*, Q^*$, there is $E_{sq,par}^{cc}(\rho,{\cal M}(\hat{x},\hat{y}))$ is a lower bound to ${\cal I}(\omega^*,Q^*,\hat{x},\hat{y})$. To this end, let us assume that the infimum in ${\cal I}$  is achieved on for a pair $(\widetilde{\sigma}_{ABE},\widetilde{\mc{M}}(x,y))\in \hat{\Sigma}(\omega^*,Q^*)$.  We then observe that for $\widetilde{\rho}_{AB} : = \tr_E \widetilde{\sigma}_{ABE}$, there is:

\begin{align}
    &{\cal I}_{par}(\omega^*,Q^*,\hat{x},\hat{y})=I(A:B\downarrow E)_{ (\widetilde{\sigma}_{ABE},\widetilde{\mc{M}}(\hat{x},\hat{y}))}  \\
    &\geq E^{cc}_{sq}(\widetilde{\rho}_{AB}, \widetilde{\mc{M}}(\hat{x},\hat{y}))  \\
    &\geq \inf_{\substack{\omega(\sigma,{\cal N})=\omega(\widetilde{\rho},\widetilde{\mc{M}})\\ P_{err}(\sigma,{\cal N}(\hat{x},\hat{y}))=P_{err}(\widetilde{\rho},\widetilde{\mc{M}}(\hat{x},\hat{y})) }}E^{cc}_{sq}(\sigma,{\cal N}(\hat{x},\hat{y}))\\
    & \geq \inf_{\substack{\omega(\sigma,{\cal N})=\omega(\rho,{\cal M})\\ P_{err}(\sigma,{\cal N}(\hat{x},\hat{y}))=P_{err}(\rho,{\cal M}(\hat{x},\hat{y})) }}E^{cc}_{sq}(\sigma,{\cal N}(\hat{x},\hat{y}))
    \\ & =
    E^{cc}_{sq,par}(\rho,{\cal M}(\hat{x},\hat{y}))
    ,\label{eq:al-1}
\end{align}
where the first inequality comes from
the fact, that channels $\Lambda:E\rightarrow E'$ in definition of $E_{sq}^{cc}$ are acting on a purification of $\widetilde{\rho}_{AB}=\tr_E\widetilde{\sigma}_{ABE}$ hence can
achieve lower value than channels acting
on system $E$ of $\widetilde{\sigma}_{ABE}$.
The next inequality is just by taking infimum, while the last is due to the 
fact that $(\rho_{AB},{\cal M}),(\widetilde{\rho}_{AB},\widetilde{\mc{M}})\in \hat{\Sigma}(\omega^*,Q^*)$.

We prove now the converse inequality.
Let $(\sigma_{AB},{\cal N})$ be a pair
achieving infimum in definiton of the 
$E_{sq,par}^{cc}(\rho,{\cal M}(\hat{x},\hat{y}))$. In particular there is $\omega(\sigma,{\cal N}) =\omega(\rho,{\cal M})=\omega^*$ by assumption and $P_{err}(\sigma,{\cal N})=P_{err}(\rho,{\cal M})=Q^*$. And hence $(\psi^{\sigma},{\cal N})\in\hat{\Sigma}(\omega^*,Q^*)$. We have then
\begin{align}
&E_{sq,par}^{cc}(\rho,{\cal M}(\hat{x},\hat{y}))=E^{cc}_{sq}(\sigma,{\cal N}(\hat{x},\hat{y}))\\
&= I(A:B\downarrow E)_{(\psi^{\sigma},{\cal N}(\hat{x},\hat{y}))}  \\
&\geq \inf_{\sigma' \in \hat{\Sigma}(\omega^*,Q^*)}I(A:B\downarrow E)_{\sigma'(\hat{x},\hat{y})} \\
&={\cal I}_{par}(\omega^*,Q^*,\hat{x},\hat{y}),
\end{align}
hence the assertion follows.
\end{proof}

In Ref.~\cite{AL20} (see Eq.~(18) there) there is also defined a quantity which is
equivalent to $E_{sq,dev}^{cc}$. It reads in our notation
\begin{equation}
    \mc{I}_{dev}(p(ab|\hat{x}\hat{y}))\coloneqq \inf_{\sigma \in {\Sigma}(p(ab|\hat{x}\hat{y}))}I(A;B\downarrow E)_{\sigma(\hat{x},\hat{y})}
\end{equation}
where $\sigma_{ABE} \in \Sigma(p(ab|\hat{x}\hat{y}))$ iff there exists a measurement ${\cal M}$ such that $\tr_E (\sigma_{ABE} {\cal M}\otimes \id_E) = p(ab|xy)$. Analogous proof to the above, with $\hat{\Sigma}$ replaced by $\Sigma$ and optimization over $\omega$ and $P_{err}$ reduced to optimization over compatible devices, leads to the following equivalence:
\begin{theorem}
    For any quantum realization $(\rho,{\cal M})$ of a device $p(ab|xy)$ and pair of inputs $(\hat{x},\hat{y})$, there is
    \begin{equation}
        {\cal I}_{dev}(p(ab|\hat{x}\hat{y})) = E_{sq,dev}^{cc}(\rho,{\cal M}(\hat{x},\hat{y})).
    \end{equation}
\end{theorem}
\color{black}

We now observe that the upper bounds
{\it plotted} in the Refs.~\cite{AL20,FBL+21} are upper bounds on $E_{sq,par}^{cc}(\rho_{AB},{\cal M}(\hat{x},\hat{y}))$.
We denote the plotted functions as $I_{AL}(\rho_{AB},{\cal M}(\hat{x},\hat{y}))$ and $I_{FBJL+}(\rho_{AB},{\cal M}(\hat{x},\hat{y}))$ respectively. That means, if $E_{sq,par}^{cc}$ was plotted, it would be below both the bounds given in these articles. 

\begin{theorem}\label{thm:main-1}
For any  device $(\rho,{\cal M})$ and input ${\cal M}(\hat{x},\hat{y})$, there is
\begin{align}
    &E_{sq,par}^{cc}(\rho_{AB},{\cal M}(\hat{x},\hat{y}))\leq I_{AL}(\rho_{AB},{\cal M}(\hat{x},\hat{y})),\label{eq:AL2}\\
    &E_{sq,dev}^{cc}(\rho_{AB},{\cal M}(\hat{x},\hat{y}))\leq I_{FBJL+}(\rho_{AB},{\cal M}(\hat{x},\hat{y})),\\
    &E_{sq,par}^{cc}(\rho_{AB},{\cal M}(\hat{x},\hat{y}))\leq 
    E_{sq,dev}^{cc}(\rho_{AB},{\cal M}(\hat{x},\hat{y})).
\end{align}
\label{thm:below}
\end{theorem}
\begin{proof}
Let us first note that the last inequality from the above follows from the fact, that the set over which infimum is taken in definition of $E_{sq,dev}^{cc}$ is contained 
in the set over which infimum is taken in definition of $E_{sq,par}^{cc}$.

To obtain \eqref{eq:AL2}, observe that the bound calculated in Ref.~\cite{AL20} is 
\begin{equation}
    I(A;B|E)_{M_A\otimes M_B\otimes \textrm{id}_E\psi^{\sigma}_{ABE}},
\end{equation}
where $M_A=\sigma_z$, $M_B=\sigma_z$ with probability $1-2P_{err}$ and a random bit with probability $2P_{err}$. Here, $\psi^{\sigma}$ is the purification of the state 
\begin{align}
    \sigma &= \frac{1+C}{2}\op{\Phi^{+}}_{AB}+ \frac{1-C}{2}\op{\Phi^{-}}_{AB},\label{eq:choice-of-state}\\
    \ket{\Phi^{\pm}}_{AB} &= \frac{1}{\sqrt{2}}\left(\ket{00}\pm\ket{11}\right),\\
    C&=\sqrt{\left(\frac{\omega}{2}\right)^2-1}
\end{align}
We then obtain the following set of inequalities
\begin{align}
      &E_{sq,par}^{cc}(\rho,{\cal M}(\hat{x},\hat{y}))\nonumber\\
     & \equiv\inf_{\underset{P_{err}(\sigma,{\cal N})= P_{err}(\rho,{\cal M})}{\omega(\sigma,{\cal N})=\omega(\rho,{\cal M})}}\inf_{\Lambda:E\rightarrow E'} I(A:B|E')_{{\mathrm M}_{AB}\otimes \Lambda_E \psi^{\sigma}_{ABE}}\\
     &\leq I(A;B|E)_{M_A\otimes M_B\otimes \textrm{id}_E\psi^{\sigma}_{ABE}}
\end{align}
This follows by choosing particular strategies as specified above, where we choose the state given in \eqref{eq:choice-of-state}, and by choosing $\Lambda_{E\rightarrow E'}$ as identity. 

\iffalse
Alternatively,the plot $I_{AL}$ can be explained as the following: first in Ref.~\cite{AL20} only single measurement $(\hat{x},\hat{y})$ is considered for key generation. Second, the bound reads $I(A:B|E)$ computed on the purification of the state which 
satisfies the same values of CHSH and QBER as the honest state under honest measurements, which we denote here as $(\rho,{\cal M})$. Hence $I_{AL}$ can be only larger than $E_{sq}^{cc}(\sigma,{\cal N}(\hat{x},\hat{y}))$ where additional infimum over channels acting on system $E$ is taken of the same quantity. We have further  $E_{sq}^{cc}(\sigma,{\cal N}(\hat{x},\hat{y}))\geq E_{sq,par}^{cc}(\rho,{\cal M})$ by definition of the $E_{sq,par}^{cc}(\rho,{\cal M})$, which proves Eq. (\ref{eq:AL}).
\fi

For the plot $I_{FBJL+}$,  we first note that the FBJ\L{}KA bound~\cite{FBL+21} works for the protocols where the measurements are projective and announced after the protocol. If we then fix a single measurement $(\hat{x},\hat{y})$, we can consider this to be known to Eve. In principle, in this case, Alice and Bob need not to announce the test rounds, as they can use sublinear amount of private key needed for authentication to encrypt this information. However whenever test rounds (and so key rounds) are available to Eve, she can measure all her shares
as if they were key rounds. This strategy will lead to the same bound
as if Alice and Bob publicly announced testing and/or key rounds.

The device $p(ab|xy)$ against which the honest parties (implicitly) perform test in Ref.~\cite{FBL+21}, is quantum, hence it is expected to be $(\rho,{\cal M})$ for some honest realization via measuring ${\cal M}$ on $\rho$ promised by provider (e.g. a Werner state of some dimension, and the measurements of CHSH inequality). But it can be in fact equal to any $(\sigma,{\cal N})$
such that $(\sigma,{\cal N})=(\rho,{\cal M})$. The idea of Ref.~\cite{FBL+21} is to represent
the device $p(ab|xy)$ as
convex combination of local and non-local part, where local part
is a mixture of local conditional distributions, i.e., 
\begin{align}
   % &(\rho,{\cal M}) = \nonumber\\
    p(ab|xy)=\sum_{i=0}^{k-1} p_i P_{L}^{(i)}(ab|xy)+ q P_{NL}(ab|xy)
    \label{eq:ccattack}
\end{align}
where $\sum_{i=0}^{k-1} p_i +q =1$ for some natural $k$, $\{P_{L}^{(i)}(ab|xy)\}_i$ is a set of local conditional distributions and $P_{NL}(ab|xy)$ is nonlocal part of the device $p(ab|xy)$. In what follows, for the clarity of argument we will first assume that $P_L^i$ are deterministic. That is, for every $i$ and $x,y$ $P_D^{(i)}(ab|xy) = \delta_{(a,b),(a_i^x,b_i^y)}$. That is, for every input the outcomes are $(a_i^x,b_i^y)$ with probability $1$.
We further relax this assumption to local distributions in Remark \ref{rem:extended_proof}.

Since the devices
in the above convex combination are
quantum, they admit quantum representation so that there exist collections of measurements ${\cal N}_L^{(i)}= \{{\mathrm N}^{x,(i)}_a\otimes {\mathrm N}^{y,(i)}_b\}_{x,y}$ and ${\cal N}_{NL}=\{{\mathrm N}^{x,(NL)}_a\otimes {\mathrm N}^{y,(NL)}_b\}_{x,y}$ and states $\sigma_i$ as well as $\sigma_{NL}$ such, that 
\begin{align}
    &P_{D}^{(i)}(ab|xy) = \Tr {\cal N}_L^{(i)} \sigma_i \nonumber\\
    &P_{NL}(ab|xy) = \Tr {\cal N}_{NL} \sigma_{NL}
\end{align}
We can then define a strategy, which realizes splitting 
of a device $p(ab|xy)$ into the above devices. To this end let us define
\begin{equation}
    \sigma_{ABA'B'}=\sum_{i=0}^{k-1} p_i \sigma_i \otimes |ii\>\<ii|_{A'B'} + q \sigma_{NL}|kk\>\<kk|_{A'B'} 
\end{equation}
and
\begin{align}
    {\cal N} =\sum_{i=0}^{k-1} {\cal N}_L^{(i)}\otimes |ii\>\<ii|_{A'B'} + {\cal N}_{NL}\otimes |kk\>\<kk|_{A'B'}
    \label{eq:calN}
\end{align}
By definition there is $(\sigma,{\cal N}) =\Tr\left(\sum_i p_i{\cal N}_L^{(i)}\sigma_{L}^{i}+q{\cal N}_{NL}\sigma_{NL}\right) = \sum_ip_i P_{D}^{(i)}(ab|xy) + qP_{NL}(ab|xy) =({\cal M},\rho)$.

We are ready to define an {\it extension} of the state $\sigma_{ABA'B'}$ to
systems $EE'$ of Eve, which
realizes distribution $p(abe|x,y)$ as defined  in \cite[Eq.~(3)]{FBL+21}, given Eve learns $(x,y)$.
\begin{align}
    &\sigma_{ABA'B'EE'}=\sum_{i=0}^{k-1} p_i \sigma_i \otimes |ii\>\<ii|_{A'B'}\otimes 
\sigma_i^{E}\otimes |i\>\<i|_{E'} \nonumber\\+  & q\sigma_{NL}\otimes |kk\>\<kk|_{A'B'}\otimes |?\>\<?|_E\otimes|k\>\<k|_{E'},
\end{align}
where $\sigma^E_i=\sigma_i$ for all $i\in\{0,...,k-1\}$.

Given the system $E'$ is in state $\op{i}$ with $i\in \{0,...,k\}$, the state of Alice and Bob collapses
to $\sigma_{AB}^k$ or $\sigma_{NL}$ respectively. Then, either $i=k$ so she
learns $|?\>\<?|$ i.e. nothing from $E$, or $i <k$ and Eve 
measures $\sigma_i^E = \sigma_i$  according to ${\cal N}_L^{(i)}(x,y)$ and learn the (deterministic) outputs
of Alice and Bob $a_i^x,b_i^y$. Note here, that due to the fact that outputs of Alice and Bob are deterministic, Eve can learn them from a copy of the state $\sigma_i$, given she performs the same measurement as they do.

In particular, if the key-generation input is single, equal to $(\hat{x},\hat{y})$, Alice, Bob and Eve can generate from $\sigma_{ABA'B'EE'}$ a distribution $p(abe|\hat{x}\hat{y})$, where $e \in {\cal A}\times {\cal B}\cup \{?\}$, where ${\cal A}$  and ${\cal B}$ are the alphabets of outputs of Alice and Bob's device given input $(\hat{x},\hat{y})$. Further, as it is proposed in Ref.~\cite{FBL+21}, depending on the state of the system $E'$, Eve applies a particular post-processing map $\Lambda^{post}_{E|E'}:E\rightarrow E''$ on her classical outputs $\delta_{(a^x_i,b^y_i),e}$ and symbol "$?$" mapping them to symbols $\{\bar{e}\}$ in order to minimize the value of $I(A:B|E'')_{p(ab\bar{e}|\hat{x}\hat{y})}$ on such obtained distribution $p(ab\bar{e}|\hat{x}\hat{y})$.

It is known, that any extension of a bipartite quantum state can be obtained by a CPTP map applied to its purifying system. Hence, there exists a map $\Lambda^{ext}_E$ which produces from a purification of $\sigma_{AB}$ denoted as $|\psi^{\sigma}_{ABE}\>$, the extension in state $\sigma_{ABA'B'EE'}$. This map composed with the measurement $\mathcal{N}_L^i(\hat{x},\hat{y})$ on the Eve system, followed by  $\Lambda_{E|E'}^{post}$ and tracing out register $E'$ results in
 desired final distribution $p(ab\bar{e}|\hat{x}\hat{y})$. 

% \textcolor{brown}{Prior to performing $\Lambda_{E|E'}^{\textrm{post}}$, we need to perform the measurements on the Eve system?Also, the measurements are dependent on the outcome of the E' system, so can we bypass performing $\Lambda^{\textrm{gen}}$?}

To summarize, the distribution $p(ab\bar{e}|\hat{x}\hat{y})$ can be obtained by applying  ${\cal N}(\hat{x},\hat{y})$ on systems AB and $\Lambda^{tot}_E:=\Tr_{E'}\circ\Lambda_{E|E'}^{post}\circ {\cal N}_L(\hat{x},\hat{y})_{E|E'}\circ \Lambda^{ext}_E$ on system $E$ of the purification $|\psi^{\sigma}_{ABE}\>$ of $\sigma_{AB}$. Here by ${\cal N}_L(\hat{x},\hat{y})_{E|E'}$ we mean that given $E'$ is in state $|i\>$, ${\cal N}_L^{(i)}(\hat{x},\hat{y})$ is measured on system $E$. We thus have:
\begin{align}
I(A:B|E'')_{p(ab\bar{e}|\hat{x}\hat{y})} &=
     I(A:B|E)_{{\cal N}(\hat{x},\hat{y})\otimes \Lambda^{tot}_E|\psi^\sigma_{ABE}\>} \nonumber\\
     &\geq
    \inf_{\Lambda:E\rightarrow E''}I(A:B|E'')_{{\cal N}(\hat{x},\hat{y})\otimes \Lambda|\psi^\sigma_{ABE}\>}\nonumber\\&=E_{sq}(\sigma,{\cal N}(\hat{x},\hat{y}))\nonumber
    \\&\geq
    \inf_{(\sigma,{\cal N})=(\rho,{\cal M})}
    E_{sq}(\sigma,{\cal N}(\hat{x},\hat{y}))\nonumber\\& \equiv
    E_{sq,dev}^{cc}(\rho_{AB},{\cal M}(\hat{x},\hat{y})).
\end{align}
It suffices to note that the plot of $I_{FBJL+}$ visualises the values of the function $I(A:B|E'')_{p(ab\bar{e}|\hat{x}\hat{y})}$ attained
on the distribution $p(ab\bar{e}|\hat{x},\hat{y})$.
Hence due to the above inequalities
we obtain 
\begin{equation}
    I_{FBJL+}(\rho,{\cal M}(\hat{x},\hat{y}))\geq E_{sq,dev}^{cc}(\rho_{AB},{\cal M}(\hat{x},\hat{y}))
\end{equation}
\end{proof}

Two remarks are in due. Both showing how to fit our approach to {\it exactly} reproduce results of Ref.~\cite{FBL+21} (however not necessarily in optimal way with respect to finding upper bounds on the key rate). We first extend the above proof to the case of splitting into local rather than deterministic devices.

\begin{remark}
In the case where the devices in Eq. (\ref{eq:ccattack}) are
not deterministic, one can explicitly
specify $\sigma_i$ and ${\cal N}_L^{(i)}$ as it is explained below,
with all other parts of the proof of Theorem~\ref{thm:below} unchanged.
For each $i \in \{0,...,k-1\}$, there exists a splitting of $P_L^{(i)}(ab|xy)$ into deterministic devices 
\begin{equation}
P_L^{(i)}(ab|xy) = \sum_{j} q^{(i)} P_D^{(ij)}(a|x)P_D^{(ij)}(b|y).
\end{equation}
We can then explicitly realize the deterministic devices as
\begin{equation}
    P_D^{(ij)}(a|x) = Tr \sigma^{(ij)}_A {\cal N}^{x}_{a,L}
\end{equation}
where $\sigma^{(ij)}_A = \otimes_{l=1}^{|{\cal X}|} |a^{(ij)}_l\>\<a^{(ij)}_l|_{A_l}$, and
\begin{align}
&{\cal N}^{x}_{a,L} =\{ {\mathrm P}_{A_x}\otimes \id_{A_{l \neq x}}\},    \label{eq:Nxa}\\
&{\cal N}^{y}_{b,L} =\{ {\mathrm P}_{B_y}\otimes \id_{B_{l\neq y}}\}, \label{eq:Nyb}
\end{align}
where ${\mathrm P}_{l}$ projects system $A_l$ (or $B_l$ respectively) onto computational basis. Having defined analogously $P_D(b|y)$, we can define the state $\sigma_{ABA'B'}$ as follows 

\begin{align}
    &\sigma_{ABA'B'}= 
    \sum_{i=0}^{k-1}p_i\left(\sum_j  q_j^{(i)}\sigma^{(ij)}_A\otimes \sigma^{(ij)}_B\right)\otimes|ii\>\<ii|_{A'B'}\nonumber\\
    &+  q \sigma_{NL}\otimes |kk\>\<kk|_{A'B'},
    \label{eq:sigma_for_extension}
\end{align}
where $\sum_j  q_j^{(i)}\sigma^{(ij)}_A\otimes \sigma^{(ij)}_B=: \sigma_i$.

With 
\begin{equation}
{\cal N}^{(i)}_L(x,y) = {\cal N}^{x}_{a,L}\otimes {\cal N}_{b,L}^{y}
\label{eq:NxaNyb}
\end{equation} for all $i\in\{0,...,k-1\}$ already defined, and ${\cal N}$ defined as in Eq. (\ref{eq:calN}), we have again $(\sigma_{ABA'B'},{\cal N})= \Tr(\sum_i p_i{\cal N}^{(i)}_L\sigma_i + q {\cal N}_{NL}\sigma_{NL}) =
({\cal M},\rho)$.

We can also define extension of the  state $\sigma_{ABA'B'}$ to Eve's systems $E_AE_B$ as shown below.
\begin{align}
    &\sigma_{ABA'B'E_AE_BE'}=\sum_{i=0}^{k-1}\sum_j p_i q_j^{(i)}\sigma^{(ij)}_A\otimes \sigma^{(ij)}_B\otimes|ii\>\<ii|_{A'B'}\nonumber\\
    &\otimes \sigma^{(ij)}_{E_A}\otimes \sigma^{(ij)}_{E_B}\otimes |i\>\<i|_{E'} +  q \sigma_{NL}\otimes |kk\>\<kk|_{A'B'} \nonumber\\ & \otimes|?\>\<?|_{E_A}\otimes 
    |?\>\<?|_{E_B} \otimes |k\>\<k|_{E'},
    \label{eq:big_extension}
\end{align}
where $\sigma^{(ij)}_{E_A} = \sigma^{(ij)}_A$ and $\sigma^{(ij)}_{E_B} = \sigma^{(ij)}_B$. Note, that given knowledge of $(x,y)$ Eve can measure ${\cal N}^{(i)}_{L}(x,y)$ on her systems $E_AE_B$ and learn the outcomes of Alice and Bob. 
% Via the system $E_r$ she has additional knowledge from which deterministic device Alice and Bob take outcomes.

We have therefore specified a tripartite quantum state, from which Alice Bob and Eve generate the distribution $p(abe|xy)$ as it is specified in \cite[Eq.~(3)]{FBL+21}. In this distribution Eve is fully correlated to the outcomes of local devices, and is fully uncorrelated (having symbol "?") with the non-local device. The remaining
part of the proof of Theorem~\ref{thm:below} is the same as shown before.
\label{rem:extended_proof}
\end{remark}
In the remark below we argue that our
approach presented in Theorem~ \ref{thm:below} is slightly more general than that of Ref.~\cite{FBL+21}.
\begin{remark}
In fact the register $E'$ is not used 
in Ref.~\cite{FBL+21}. There, the distribution $P(ab\bar{e}|\hat{x}\hat{y})$ depends
only from the outputs $(a,b)$ and $"?"$ and {\it not} on
the number of a deterministic device 
that produces this output. The system $E'$ appeared in our discussion as a mean to realize the condition of Ref.~\cite{FBL+21} that Eve should obtain the outputs of Alice and Bob in case when the device shared by them is local. Whether one can achieve this goal without additional information held by the index $i$ is possible, we leave as an open problem. We also keep system $E'$ and its use in the description (proof of Theorem~\ref{thm:below}) due to the fact that it shows that Eve has more knowledge, that
may lead to potentially tighter upper bounds.
\label{rem:no_eprime}
\end{remark}

As a corollary there comes the following fact:
\begin{corollary}\label{cor:main-1}
For any device $(\rho,{\cal M})$ and a pair of inputs generating the key $(\hat{x},\hat{y})$, there is:
\begin{align}
    &K^{iid,(\hat{x},\hat{y})}_{DI,par}(\rho,{\cal M}) \leq 
       E_{sq,par}^{cc}(\rho,{\cal M}(\hat{x},\hat{y})) \\
       &\leq\min \{I_{AL}(\rho,{\cal M}(\hat{x},\hat{y})),I_{FBJL+}(\rho,{\cal M}(\hat{x},\hat{y}))\}
    \end{align}
    \label{cor:below}
\end{corollary}
\begin{proof}
It holds due to the Corollary~\ref{col:esqbound} and the 
Theorem~\ref{thm:below}.
\end{proof}
To state the main theorem of this section it suffices to argue
that $E_{sq,par}^{cc}(\rho,{\cal M}(\hat{x},\hat{y}))$ is convex.
\begin{lemma}
The $E_{sq,par}^{cc}$ is convex i.e.
for every device $(\bar{\rho},{\cal M})$
and an input pair $(\hat{x},\hat{y})$
there is
\begin{align}
    &E_{sq,par}^{cc}(\bar{\rho},{\cal M}(\hat{x},\hat{y})) \leq \\ 
    &p_1E_{sq,par}^{cc}(\rho_1,{\cal M}(\hat{x},\hat{y})) +
    p_2 E_{sq,par}^{cc}(\rho_2,{\cal M}(\hat{x},\hat{y}))
\end{align}
where $\bar{\rho}=p_1\rho_1+ p_2 \rho_2$ and $p_1+p_2=1$ with $0\leq p_1 \leq 1$.
\label{lem:conv2}
\end{lemma}
\begin{proof}
Let us fix two strategies 
$(\sigma_1,{\cal N}_1)$ and
$(\sigma_2,{\cal N}_2)$ such that
$\omega(\sigma_i,{\cal N}_i)=\omega(\rho_i,{\cal M})$ and
$P_{\mbox{err}}(\sigma_i,{\cal N}_i)=P_{\mbox{err}}(\rho_i,{\cal M})$.
Consider also a state $\bar{\sigma} =
p_1\sigma_1\otimes \op{00}_{A'B'} +
p_2 \sigma_2 \otimes |11\>\<11|_{A'B'}$, and a joint strategy
${\cal N}= {\cal N}_1\otimes \op{00}_{A'B'} + {\cal N}_2\otimes |11\>\<11|_{A'B'}$. We note then that by linearity of $\omega$, there is:
\begin{align}
    &\omega(\bar{\sigma},{\cal N}) =
    \omega(\sum_i p_i \Tr{\cal N}_i\sigma_i) = \\
    &\sum_i p_i  \omega(\sigma_i,{\cal N}_i)=\sum_ip_i\omega(\rho_i,{\cal M})=\omega(\bar{\rho},{\cal M}),
    \end{align}
    where in the pre-last equality we have used the fact that strategies $(\sigma_i,{\cal N}_i)$ reproduces statistics of $(\rho_i,{\cal M})$ respectively.
    Analogously we obtain:
    \begin{align}
&P_{\mbox{err}}(\bar{\sigma},{\cal N}) =
    P_{\mbox{err}}(\bar{\rho},{\cal M})
\end{align}
This implies that
\begin{equation}
    E_{sq,par}^{cc}(\bar{\rho},{\cal M}(\hat{x},\hat{y})) \leq E_{sq}^{cc}(\bar{\sigma},{\cal N}(\hat{x},\hat{y}))
\end{equation}
since infimum over strategies is less than the value of the function taken in particular strategy described above by $(\bar{\sigma},{\cal N})$.
We use further convexity of the $E_{sq}^{cc}$ function proved in Lemma~\ref{lem:conv} to get
\begin{align}
&E_{sq,par}^{cc}(\bar{\rho},{\cal M}(\hat{x},\hat{y})) \leq p_1 E_{sq}^{cc}(\sigma_1\otimes \op{00}_{A'B'},{\cal N}(\hat{x},\hat{y})) +    \nonumber\\
& p_2 E_{sq}^{cc}(\sigma_2\otimes |11\>\<11|_{A'B'},{\cal N}(\hat{x},\hat{y}))
\end{align}

We further note that by definition of ${\cal N}$ there is:
\begin{equation}
    E_{sq}^{cc}(\sigma_1\otimes \op{00}_{A'B'},{\cal N}(\hat{x},\hat{y})) = 
    E_{sq}^{cc}(\sigma_1,{\cal N}_1(\hat{x},\hat{y})).
\end{equation}
Indeed, ${\cal N}(\hat{x},\hat{y})(\sigma_1\otimes \op{00}) = {\cal N}_1(\hat{x},\hat{y})\sigma_1\otimes \op{00}$. 
Below we have a slight change in notation. From here on instead of $I(A:B|E)_{\rho}$, we use $I(A;B|E)[\rho]$. We also represent the purification $\psi_{ABE}^{\sigma}$ as $\psi^{ABE}(\sigma)$.
Hence, denoting by $\psi(\sigma)$ a purification of a state $\sigma$ we get: 
\begin{align}
    &E_{sq}^{cc}(\sigma_1\otimes \op{00}_{A'B'},{\cal N}(\hat{x},\hat{y})) =\nonumber\\
    &E_{sq}^{cc}(\sigma_1\otimes \op{00}_{A'B'},{\cal N}_1(\hat{x},\hat{y})\otimes \id_{A'B'}) = \nonumber\\
    &\inf_{\Lambda: E\rightarrow E'}I(AA':BB'|E')[({\cal N}_1(\hat{x},\hat{y})\otimes \id_{A'B'}\otimes \Lambda_E)\nonumber \\
    &\psi^{ABA'B'E}{(\sigma_1\otimes \op{00}_{A'B'})}] =\nonumber\\
    &\inf_{\Lambda: E\rightarrow E'}I(AA':BB'|E')[({\cal N}_1(\hat{x},\hat{y})\otimes \id_{A'B'}\otimes \Lambda_E)\nonumber\\
    &\psi^{ABE}({\sigma_1})\otimes \op{00}_{A'B'})] \nonumber\\
    &=\inf_{\Lambda: E\rightarrow E'}I(A:B|E')[{\cal N}_1(\hat{x},\hat{y})\otimes \Lambda_E\psi^{ABE}({\sigma_1})]=\nonumber\\
    &E_{sq}^{cc}(\sigma_1,{\cal N}(\hat{x},\hat{y})).
\end{align}
In the second last equality we have used the fact that $I(A:B|E')[\rho_{ccq}\otimes \op{00}_{A'B'}]$ with $\rho_{ccq} := ({\cal N}_1(\hat{x},\hat{y})\otimes \id_{A'B'}\otimes \Lambda_E)
    \psi^{ABE}({\sigma_1})$ equals
    just $I(A:B|E')[\rho_{ccq}]$ since the pure state $\op{00}_{A'B'}$ does not alter the von-Neumann entropies involved in definition of the conditional mutual information.
 Similarly
 \begin{equation}
    E_{sq}^{cc}(\sigma_2\otimes |11\>\<11|_{A'B'},{\cal N}(\hat{x},\hat{y})) = 
    E_{sq}^{cc}(\sigma_2,{\cal N}_2(\hat{x},\hat{y}))
\end{equation}
Hence, there is
\begin{align}
    E_{sq,par}^{cc}(\bar{\rho},{\cal M}(\hat{x},\hat{y})) \leq & 
    p_1E_{sq}^{cc}(\sigma_1,{\cal N}_1(\hat{x},\hat{y}))+\nonumber\\
    &\quad  p_2E_{sq}^{cc}(\sigma_2,{\cal N}_2(\hat{x},\hat{y})).
\end{align}
Now, since strategies $(\sigma_i,{\cal N}_i)$ were arbitrary within their constraints, we obtain:
\begin{align}
    &E_{sq,par}^{cc}(\bar{\rho},{\cal M}(\hat{x},\hat{y})) \nonumber\\&\leq 
    p_1 E_{sq,par}^{cc}(\sigma_1,{\cal N}_1(\hat{x},\hat{y}))+
     p_2 E_{sq,par}^{cc}(\sigma_2,{\cal N}_2(\hat{x},\hat{y})),
\end{align}
hence the assertion follows.
\end{proof}
We are ready to state the main Theorem of this section. In what follows we narrow considerations to  $({\cal M},(\hat{x},\hat{y}))$ being {\it projective}, as the bound for Werner states presented in Ref.~\cite{FBL+21} applies only to this case. 

\begin{theorem}
For a Werner state $\rho_{AB}^W$
and ${\cal M}$ consisting of  projective measurements ${\mathrm P}_a^x\otimes {\mathrm P}_b^y$, and a pair of inputs $(\hat{x},\hat{y})$ used to generate the key, there is
\begin{align}
    &K_{DI,par}^{iid,(\hat{x},\hat{y})}(\rho_{AB}^W,{\cal M})\leq \nonumber\\ &\mbox{Conv}(I_{AL}(\rho_{AB}^W,{\cal M}(\hat{x},\hat{y})),I_{FBJL+}(\rho_{AB}^W,{\cal M}(\hat{x},\hat{y}))),
\end{align}
where $Conv(P_1,P_2)$ is the convex hull of the plots of functions $P_i$, and $K_{DI,par}^{iid,(\hat{x},\hat{y})}(\rho_{AB}^W,{\cal M})$ is defined with respect to $\omega = CHSH$ and $P_{err}=P(a\neq b|\hat{x}\hat{y})$.
\label{thm:convexification}
\end{theorem}
\begin{proof}
For the proof it suffices to note that by Corollary~\ref{cor:below} we have
\begin{align}
    &K^{iid,(\hat{x},\hat{y})}_{DI,par}(\rho^W_{AB},{\cal M})  \nonumber\\ &\leq E_{sq,par}^{cc}(\sigma,{\cal M}(\hat{x},\hat{y}))\nonumber\\ &\leq \min\{I_{AL}(\sigma_{AB},{\cal M}(\hat{x},\hat{y})),
    I_{FBJL+}(\sigma_{AB},{\cal M}(\hat{x},\hat{y}))\}
\end{align}
Now, by Lemma~\ref{lem:conv2} the $E_{sq,par}^{cc}$ is convex. It is also below the plots of $I_{AL}(\sigma_{AB},{\cal M}(\hat{x},\hat{y}))$ and $I_{FBJL+}(\sigma_{AB},{\cal M}(\hat{x},\hat{y})$, due to the above inequality. As such, it must be below their convex hull. It also upper bounds the key, hence the key must be below the convex hull of the plots of $I_{AL}$ and $I_{FBJL+}$ as well.
\end{proof}

\begin{figure}
    \centering
    \includegraphics[width=1\linewidth]{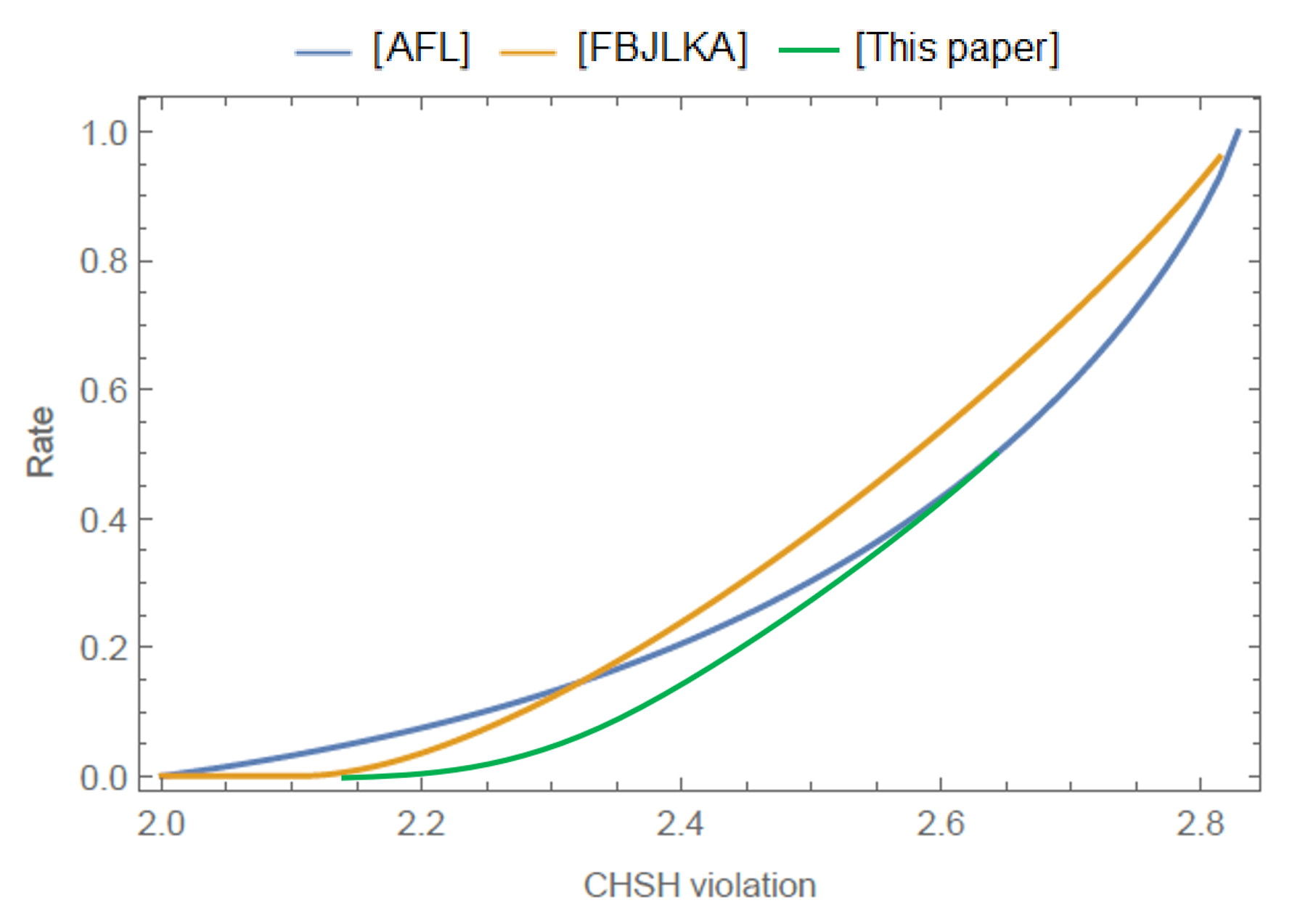}
    \caption{In this figure, we show the plots for standard device-independent CHSH protocol obtained in Refs.~\cite{AL20}, \cite{FBL+21}, and the upper bound given in Theorem~\ref{thm:convexification-m}, which is the convex hull of the former bounds, depicted in green.}
    \label{Convex-hull}
\end{figure}
\subsection{Extension to more measurements}
One can consider the function $E_{sq}^{cc}$ for multiple measurements defined as follows:
\begin{definition}
The cc-squashed entanglement of the collection of measurements ${\cal M}$ measured with distribution $p(x,y)$ of the inputs reads:
\begin{equation}
E^{cc}_{sq}(\rho_{AB},{\cal M},p(x,y)):=\sum_{x,y} p(x,y) E_{sq}^{cc}(\rho_{AB},{\mathrm M}_{x,y}).
\end{equation}
\end{definition}
Similarly to Observation \ref{obs:equiv} we have that
\begin{align}
    &E_{sq}^{cc}(\rho,{\cal M},p(x,y)) = \nonumber\\ &\sum_{x,y}p(x,y)\inf_{\rho_{ABE}=Ext(\rho_{AB})} I(A:B|E)_{{\mathrm M}_{x,y}\otimes \id_E \rho_{ABE}} 
\end{align}
We note here that the extensions $\rho_{ABE}$ can be different for different choices for $(x,y)$. 
We then note a general fact that
a convex combination of convex functions is a convex function itself.
\begin{lemma}
Let $\{f_i\}$ be the set of convex functions. Then for every distribution $\{p_i\}$ the function $\sum_i p_i f_i$ is convex.
\end{lemma}
\begin{proof}
Let $x = p x_1 +(1-p)x_2$ then,
\begin{align}
\sum_i p_i f_i (x)& \leq \sum_i p_i (p f_i (x_1) + (1-p) f_i(x_2))  \nonumber\\
&=p\sum_i p_i f_i(x_1) + (1-p)\sum_i p_i f_i(x_2).
\end{align}
\end{proof}
From the above lemma it follows that
due to convexity of $E_{sq}^{cc}(\rho,{\mathrm M})$ the
function $E_{sq}^{cc}(\rho, {\cal M},p(x,y))$ is convex. Further, due to convexity of the latter function we have that the 
analogously defined reduced version of this function 
\begin{equation}
    E_{sq,dev}^{cc}(\rho, {\cal M},p(x,y)):=
    \inf_{{(\sigma,{\cal N})=(\rho,{\cal M})}}E_{sq}^{cc}(\sigma,{\cal N},p(x,y))
\end{equation}
is also convex (via analogous lemma 
to \ref{lem:conv2}).

It will appear crucial to notice,
that in DI QKD it is assumed, that the distribution of inputs $p(x,y)$
is drawn from a private shared randomness held by Alice and Bob, which is independent of the device $(\rho,{\cal M})$ (In most cases $p(x,y)$ is the uniform distribution. Otherwise sharing private correlations in order to choose inputs based on these correlations would imply sharing private key. It would be then no sense to run a DI QKD, given Alice and Bob already share the key in form of these correlations). Due to this ``free will'' assumption, it is not known to Eve for each run which $(x,y)$ was chosen by Alice and Bob. This means that a priori Eve does not
have access to systems $E_xE_y$ of an extension of the form
\begin{equation}\label{eq:eve-extension-measurements}
    \sum_{x,y}p(x,y){\mathrm M}_{x,y}\id_E\rho_{ABE}^{(x,y)}\otimes |xy\>\<xy|_{E_xE_y},
\end{equation}
where $\rho_{ABE}^{(x,y)}$ is an extension of $\rho_{AB}$ for each $(x,y)$. However, under assumption that
{\it Alice and Bob make the announcements for the choice of measurements and Eve subsequently learns this measurement} \cite{FBL+21}, Eve can have access to the extensions given in \eqref{eq:eve-extension-measurements}. To obtain these extensions, we can assume that the eavesdropper can act on its quantum system by a map $\Lambda_{E\rightarrow E_xE_y}$ which is dependent on the measurements $(x,y)$.  It is crucial for further considerations that the Eve
has access to the above extension. 
To make this assumption explicit
we will consider the following QDI key rate:
\begin{align}
    &K_{DI,dev}^{iid,broad}(\rho,{\cal M},p(x,y)):= \nonumber\\
    &\inf_{\epsilon>0}\limsup_n\sup_{{\cal P}\in LOPC}\inf_{{(\sigma,{\cal N})\approx_\epsilon(\rho,{\cal M})}}\nonumber\\
     &\kappa^{\epsilon}_n({\cal P}([\sum_{x,y}p(x,y){\mathrm N}_{xy}\otimes \id_E(\psi^\sigma_{ABE}\otimes |xy\>\<xy|_{E_xE_y})]^{\otimes n})),
\end{align}
where by {\it broad} we mean that $(x,y)$ are broadcasted, and made
explicit by adding systems $E_xE_y$ to 
Eve. 

We denote the action of broadcasting the values of $(x,y)$ (creating systems $E_xE_y$) as ${\cal C}$. This allows us to state the following technical lemma:

\begin{lemma}
The function $E_{sq}^{cc}(\rho,{\cal C}\circ\sum_{x,y} p(x,y) {\mathrm M}_{x,y})$ is convex in the second argument i.e. 
\begin{equation}
    E_{sq}^{cc}(\rho,{\cal C}\circ\sum_{x,y} p(x,y) {\mathrm M}_{x,y}) \leq 
    \sum_{x,y} p(x,y) E_{sq}^{cc}(\rho,{\mathrm M}_{x,y})
\end{equation}
\label{lem:technical}
\end{lemma}
\begin{proof}
We can write
\begin{multline}
    E_{sq}^{cc}(\rho,{\cal C}\circ\sum_{x,y} p(x,y) {\mathrm M}_{x,y})=\\ \inf_{\Lambda_{EE_xE_y}}I(A;B|EE_xE_y)_{\Lambda_{EE_xE_y}(\rho_{ABEE_xE_y})}.
\end{multline}
We have constructed a particular extension of $\rho$ measured by $\sum_{x,y} p(x,y) {\mathrm M}_{x,y}$ as follows:
\begin{align}
    &\rho_{ABEE_xE_y}\coloneqq \nonumber \\ & \sum_{x,y} p(x,y) 
    {\mathrm M}_{x,y}\otimes \id_{EE_xE_y} \sigma_{ABE}\otimes |xy\>\<xy|_{E_xE_y},
\end{align}
where $\sigma_{ABE}$ is an arbitrary extension of the state $\rho$. The map $\Lambda_{EE_xE_y}$ is arbitrary. 
The access to the registers $E_xE_y$ is assured by application of a broadcasting map ${\cal C}$ after performing the measurement.
It is straightforward to see that
upon tracing out $E,E_x,E_y$ we obtain
 $\rho$ measured by
a convex combination of ${\mathrm M}_{x,y}$. 
Now, let us choose a particular map of the form
\begin{equation}
    \tilde{\Lambda}_{EE_xE_y} = \sum_{x,y}\tilde{\Lambda}_E^{x,y}\otimes \op{xy}_{E_xE_y},
\end{equation}
where $\tilde{\Lambda}_{E}^{x,y}$ is arbitrary. We then obtain
\begin{align}
    &E_{sq}^{cc}(\rho,{\cal C}\circ\sum_{x,y} p(x,y) {\mathrm M}_{x,y})\\&\leq  I(A;B|EE_xE_y)_{\tilde{\Lambda}_{EE_xE_y}(\rho_{ABEE_xE_y})}\\
    &=\sum_{x,y}p(x,y)I(A;B|E)_{M_{x,y}\otimes\tilde{\Lambda}_E^{x,y}(\sigma_{ABE})},
\end{align}
where $\tilde{\Lambda}_{EE_xE_y}(\rho_{ABEE_xE_y})= \sum_{x,y} p(x,y) 
    {\mathrm M}_{x,y}\otimes \tilde{\Lambda}_E^{x,y}\otimes\id_{E_xE_y} \sigma_{ABE}\otimes |xy\>\<xy|_{E_xE_y}$. 
Since $\tilde{\Lambda}_E^{x,y}$ is an arbitrary map, we obtain
\begin{align}
&E_{sq}^{cc}(\rho,{\cal C}\circ\sum_{x,y} p(x,y) {\mathrm M}_{x,y})\\&\leq \sum_{x,y}p(x,y)\inf_{\Lambda_E^{x,y}}I(A;B|E)_{M_{x,y}\otimes{\Lambda}_E^{x,y}(\rho_{ABE})}
\\&=\sum_{x,y}p(x,y)E_{sq}^{cc}(\rho,{\mathrm M}_{x,y})
\end{align}
\iffalse

Since the above choice is particular we have
\begin{align}
    &E_{sq}^{cc}(\rho,\mathcal{C}\circ\sum_{x,y} p(x,y) {\mathrm M}_{x,y}) \leq 
     I(A:B|EE_xE_y)[\rho_{ABEE_xE_y}]\\
    &=\sum_{x,y} p(x,y) I(A:B|E)[{\mathrm M}_{x,y}\otimes\id_E (\sigma^{(x,y)}_{ABE})].
    \label{eq:flags}
\end{align}
The above inequality follows from the following observation, (analogous to Observation~\ref{obs:equiv}),
\begin{align}
    &E_{sq}^{cc}(\rho,\mathcal{C}\circ\sum_{x,y} p(x,y) {\mathrm M}_{x,y})  =\nonumber\\ &\inf_{\sigma_{ABEE_xE_y}=Ext(\rho)}I(A:B|E)_{\sum_{x,y}p(x,y){\mathrm M}_{x,y}\otimes \id_E\sigma_{ABEE_xE_y}}.
\end{align}
Hence, fixing particular extension to be $\rho_{ABEE_xE_y}$ only increases the value.
In the equality (\ref{eq:flags}) we use the fact that conditioning the mutual information upon classical register $E_xE_y$ is equal to average conditional mutual information. Since
$\sigma^{(x,y)}_{ABE}$ were chosen arbitrarily, they can realize the infima given in def. of $E_{sq}^{cc}(\rho,{\mathrm M}_{x,y})$ respectively,
\fi
This concludes the proof.
\end{proof}
We note now,
that $E_{sq}^{cc}(\rho, {\cal M},p(x,y))$ is
an upper bound for distillable key
of the state $\sum_{x,y}p(x,y)M_{x,y}\otimes
\id_E\op{\psi^\rho}_{ABE}\otimes |xy\>\<xy|_{E_xE_y}$.

\begin{theorem}
For a bipartite state $\rho$ and a set of measurements ${\cal M}$, performed with probabilities $p(x,y)$ on it, there is

\begin{align}
&K_{DD}\left(\sum_{x,y} p(x,y){\mathrm M}_{x,y}\otimes \id_E |\psi_{\rho}\>\<\psi_\rho|\otimes |xy\>\<xy|_{E_xE_y}\right) \nonumber\\
& \leq E_{sq}^{cc}(\rho,{\cal M},p(x,y)).
\end{align}
\label{thm:Matthias_mixed_meas}
\end{theorem}
\begin{proof}
The proof follows from \cite[Theorem 3.5]{CEHHOR} for 
a tripartite ccq state $\rho_{ccq}:=\sum_{x,y}p(x,y){\mathrm M}_{x,y}\otimes \id |\psi_{\rho}\>\<\psi_\rho|\otimes |xy\>\<xy|_{E_xE_y}$, and
noticing that $K_{DD}(\rho_{ccq})\leq I(A:B\downarrow E)_{\rho_{ccq}} =E_{sq}^{cc}(\rho,{\cal C}\circ\sum_{x,y}p(x,y){\mathrm M}_{x,y})\leq \sum_{x,y}p(x,y) E_{sq}^{cc}(\rho,{\mathrm M}_{x,y}) \equiv E_{sq}^{cc}(\rho,{\cal M},p(x,y))$,
where the last inequality follows from Lemma~\ref{lem:technical}.
\end{proof}
We are ready to formulate the analogue
of the Corollary~ \ref{cor:reduced_dev}.

\begin{corollary}
The iid quantum device independent key achieved by protocols using measurements of a device $(\rho,{\cal M})$, with probability $p(x,y)$ is upper bounded as follows:    

\begin{align}
    &K_{DI,dev}^{iid,broad}(\rho,{\cal M},p(x,y))\equiv \nonumber\\
    &\inf_{\epsilon>0}\limsup_n\sup_{{\cal P}\in LOPC}\inf_{{(\sigma,{\cal N})\approx_\epsilon(\rho,{\cal M})}}\nonumber\\
    & \kappa^{\epsilon}_n({\cal P}([\sum_{x,y}p(x,y){\mathrm N}_{xy}\otimes \id_E(|\psi_\sigma\>\<\psi_\sigma|\otimes|xy\>\<xy|_{E_xE_y})]^{\otimes n}))
     \label{eq:di_par}\\ &\leq\inf_{{(\sigma,{\cal N})=(\rho,{\cal M})}}E_{sq}^{cc}(\sigma,{\cal N},p(x,y)) =: E_{sq,dev}^{cc}(\rho,{\cal M},p(x,y)),
\end{align}
where ${\mathrm N}_{xy}$ are measurements induced by $(x,y)$ on ${\cal N}$ respectively.
\label{col:esqbound2}
\end{corollary}
\begin{proof}
It follows from similar lines as the proof of the Lemma~\ref{lem:ubound_by_key} to show that

\begin{align}
    &K_{DI,dev}^{iid,broad}(\rho,{\cal M},p(x,y)) \leq \inf_{{(\sigma,{\cal N})=(\rho,{\cal M})}}\nonumber\\ & K_{DD} (\sum_{x,y}p(x,y){\mathrm M}_{x,y}\otimes \id_E |\psi_\rho\>\<\psi_\rho|\otimes |xy\>\<xy|_{E_xE_y}).
\end{align}

The assertion follows then from the Theorem ~\ref{thm:Matthias_mixed_meas}.
\end{proof}

Let us note, that the above bound
is in principle tighter than the one
considered in Ref.~\cite{FBL+21}, as it is stated in the Theorem below.

\begin{theorem}
The function $E_{sq,dev}^{cc}(\rho,{\cal M},p(x,y))$ is (i) a convex upper bound on $K_{DI,dev}^{iid,broad}(\rho,{\cal M},p(x,y))$ and (ii) a lower bound to 
the upper bound given in \cite[Eq.~(5)]{FBL+21}.
\end{theorem}
\begin{proof}
The first part of the proof follows
from the Corollary~\ref{col:esqbound2}.
The convexity of this upper bound has been already observed, as analogous to
the one of $E_{sq,par}^{cc}$ proved in the Lemma~\ref{lem:conv2}. We focus now on showing that this function is a lower bound to the upper bound given in Ref.~\cite{FBL+21}.

Let us first restrict the attacks to
such that the channel $\Lambda$ involved in
definition of the $E_{sq,dev}^{cc}(\rho,{\cal M},p(x,y))$ is a POVM i.e. has only classical outputs,
denoted as $\Lambda^{cl}_E$. In such a case we have
\begin{align}
    &E_{sq,dev}^{cc}(\rho,{\cal M},p(x,y))\leq \inf_{{(\sigma,{\cal N})=(\rho,{\cal M})}}\sum_{x,y}p(x,y)
    \nonumber\\
    &\inf_{\Lambda^{post}_{E}\circ\Lambda^{cl}_E } I(A:B|E)[{\mathrm N}_{xy}\otimes \Lambda^{post}_{E}\circ\Lambda^{cl}_E  |\psi_\sigma\>] \\
    &=\inf_{{(\sigma,{\cal N})=(\rho,{\cal M})}}  \sum_{x,y}p(x,y) \inf_{\Lambda_E^{cl}}I(A:B\downarrow E)[{\mathrm N}_{xy}\otimes \Lambda^{cl}_E |\psi_\sigma\>] \nonumber\\
    &\leq \sum_{x,y}p(x,y)I(A:B\downarrow E)[{\tilde {\mathrm N}}_{xy}\otimes {\tilde \Lambda_E^{cl}}|\psi_{\tilde{\sigma}}\>]
    \label{eq:fixing} \\
    &\equiv \sum_{p(x,y)}p(x,y)I(A:B\downarrow EE')[p(abei|xy)] \leq  \label{eq:intrinsic}\\
    & \sum_{p(x,y)}p(x,y)I(A:B\downarrow E)[p(abe|xy)],
\end{align}
where $I(A:B\downarrow E)[{p(abe|xy)}]$ is the {\it intrinsic information} of the distribution $p(abe|xy)$. 
(In the last line we have obtained the bound given in \cite[Eq.~(5)]{FBL+21}).

The first
inequality is due to restriction of the infimum to that over POVMs with classical outputs only. The first equality follows from using the definition of intrinsic information which absorbs minimization over channels $\Lambda^{post}_{E}$. The inequality (\ref{eq:fixing}) follows from (i) fixing a particular choice of the attack $({\tilde {\cal N}},\tilde{\sigma}):=({\cal N},\sigma)$, where $\sigma$ is given in Eq. (\ref{eq:sigma_for_extension}) and ${\cal N}$ is defined via (\ref{eq:Nxa}), (\ref{eq:Nyb}) and (\ref{eq:calN})  (ii)  by choosing 
$\tilde{\Lambda}^{ext}_{E}$ such that it produces extension $\sigma_{ABA'B'E_AE_BE'}$ given in Eq. (\ref{eq:big_extension}), when acting on system $E$ of $|\psi^\sigma\>_{ABE}$. 
(iii) the choice of a channel $\tilde{\Lambda}_E^{cl}:={\cal N}_L^{(i)}(x,y)_{E|E'}\circ \tilde{\Lambda}_{E}^{ext}$ where measurements ${\cal N}_L^{(i)}(x,y)={\cal N}^x_{a,L}\otimes {\cal N}^y_{b,L}$ are given in Eq. (\ref{eq:NxaNyb}). This is possible for Eve because, as it was discussed earlier, Alice and Bob broadcast the input choices $(x,y)$. This choice results
in classical systems $EE'$ holding
pairs $(e,i)$ with $e \in {\cal A}\times {\cal B} \cup \{?\}$ and $i\in \{0,...,k\}$, where $e = (a,b)$ i.e. the outputs of Alice and Bob given input $x,y$ has been chosen. We thus observe in Eq. (\ref{eq:intrinsic}), that the minimized conditional information is equal to the intrinsic information of such obtained distribution $p(abei|xy)$.

The last inequality
is due to the fact, that we first trace out register $E'$, so that the channel
involved in definition of the intrinsic information does not depend on $i$ (the information from which local device
Eve obtains the outputs). This narrows the infimum over channels in the definition of intrinsic information, hence the quantity under consideration can only go up. As a result the intrinsic information is a function of distribution $p(abe|xy)$, as it is 
obtained in \cite[Eq.~(5)]{FBL+21}. (see Remark \ref{rem:no_eprime} in this context).
\end{proof}

As the second conclusion from the above Theorem there comes the fact that 
for any family of plots of the upper bound via the average intrinsic information given in Ref.~\cite{FBL+21}, the device independent key is below their convex hull.

As we see above $E_{sq,dev}^{cc}(\rho,{\cal M},p(x,y)$ as well as intrinsic non-locality \cite{KWW20} are based on conditional mutual information where the Eve system is an extension system of underlying strategy. For completeness, we give here the definition of the quantum intrinsic non-locality as introduced in Ref.~\cite{KWW20}. 
\begin{definition}
The quantum intrinsic non-locality of a correlation $p(a,b|x,y)$ is defined as 
\begin{equation}
    N^{Q}(p(a,b|x,y))= \sup_{p(x,y)}\inf_{\rho_{\bar{A}\bar{B}XYE}}I(\bar{A};\bar{B}|XYE)_{\rho},
\end{equation}
where 
\begin{multline}
    \rho_{\bar{A}\bar{B}XYE}= \sum_{x,y,a,b}p(x,y)p(a,b|x,y)\op{a}_{\bar{A}}\otimes\op{b}_{\bar{B}}\\\otimes\op{x}_{X}\otimes \op{y}_{Y}\otimes \rho_E^{a,b,x,y}.
\end{multline}
Here, $p(a,b|x,y)\rho_E^{a,b,x,y}= \operatorname{Tr}_{AB}\left[(\Lambda^x_a\otimes\Lambda^y_b)\rho_{ABE}\right]$ and $\rho_{ABE}$ is the extension of $\rho_{AB}$. 
\end{definition}

The major differences between the two quantities is as follows: the intrinsic non-locality is a function of the device $\left\{p(a,b|x,y)\right\}$ while $E_{sq,dev}^{cc}(\rho,{\cal M}, p(x,y))$ is a function of the compatible $\rho_{ccq}$ states. For most DI-QKD protocols, the testing rounds are only important while choosing the compatible strategies, but have no further role to play in the key generation protocol. This distinction between the testing and key generation rounds can be exploited via $E_{sq,dev}^{cc}(\rho,{\cal M}, p(x,y))$ to upper bounds the key rate for protocols with \textit{specific} inputs. The presence of $p(x,y)$ in the definition of the intrinsic non-locality doesn't allow for this clear distinction of the key generation and testing rounds. Another major difference is that with $E_{sq,dev}^{cc}(\rho,{\cal M}, p(x,y))$ we allow for a flexibility on the channels that Eve can act upon her extension systems. That is, Eve's actions on the extensions can be dependent on the measurements performed by Alice and Bob. These two differences in the structure of the quantities are vital to obtain tighter bounds.

\bibliographystyle{alpha}

\bibliography{di-qkd}
\end{document}